\documentclass[11pt,twoside]{imsart}

\usepackage[hypertexnames=false,hyperfootnotes=false]{hyperref}
\usepackage{amsfonts}
\usepackage{amsthm}
\usepackage{amsmath}
\usepackage{amssymb}
\usepackage{enumerate}
\usepackage{color}
\usepackage{xspace}

\newtheorem{theorem}{Theorem}
\newtheorem{lemma}[theorem]{Lemma}

\newtheorem{proposition}[theorem]{Proposition}

\newtheorem{definition}{Definition}

\newcommand{\field}[1]{\ensuremath{\mathbb{#1}}}
\newcommand{\N}{\ensuremath{\field{N}}} 
\newcommand{\R}{\ensuremath{\field{R}}} 
\newcommand{\Rp}{\ensuremath{\R_+}} 
\newcommand{\Z}{\ensuremath{\field{Z}}} 
\newcommand{\Zp}{\ensuremath{\Z_+}} 
\newcommand{\1}{\ensuremath{\mathbf{1}}} 
\newcommand{\I}[1]{\ensuremath{\mathbb{I}_{\left\{#1\right\}}}} 
\newcommand{\tends}{\ensuremath{\rightarrow}} 
\newcommand{\PR}{\ensuremath{\mathsf{P}}} 
\newcommand{\E}{\ensuremath{\mathsf{E}}} 
\newcommand{\defeq}{\ensuremath{\triangleq}}
\newcommand{\subjectto}{\ensuremath{\mathrm{subject\ to}}} 

\newcommand{\Bscr}{\ensuremath{\mathcal B}}

\newcommand{\Dscr}{\ensuremath{\mathcal D}}
\newcommand{\Escr}{\ensuremath{\mathcal E}}
\newcommand{\Fscr}{\ensuremath{\mathcal F}}

\newcommand{\Iscr}{\ensuremath{\mathcal I}}
\newcommand{\Jscr}{\ensuremath{\mathcal J}}
\newcommand{\Kscr}{\ensuremath{\mathcal K}}
\newcommand{\Lscr}{\ensuremath{\mathcal L}}

\newcommand{\Pscr}{\ensuremath{\mathcal P}}

\newcommand{\Rscr}{\ensuremath{\mathcal R}}
\newcommand{\Sscr}{\ensuremath{\mathcal S}}

\newcommand{\Vscr}{\ensuremath{\mathcal V}}

\newcommand{\Xscr}{\ensuremath{\mathcal X}}

\newcommand{\Zscr}{\ensuremath{\mathcal Z}}
\DeclareMathOperator{\diag}{diag}

\DeclareMathOperator*{\argmax}{\mathrm{argmax}}
\newcommand{\minimize}{\ensuremath{\mathop{\mathrm{minimize}}\limits}}
\newcommand{\maximize}{\ensuremath{\mathop{\mathrm{maximize}}\limits}}
\newcommand{\emailhref}[1]{\href{mailto:#1}{\tt #1}} 

\newcommand{\bD}{{\mathbf D}}
\newcommand{\bS}{{\mathbf S}}
\newcommand{\bd}{{\mathbf d}}
\newcommand{\by}{{\mathbf y}}
\newcommand{\bdx}{{\mathbf x}}

\newcommand{\dt}{\frac{d}{dt}}

\newcommand{\bzero}{{\mathbf 0}}
\newcommand{\bone}{{\mathbf 1}}

\newcommand{\zscr}{\ensuremath{\mathfrak z}}

\newcommand{\An}{{\sf A}}
\newcommand{\an}{{\sf a}}

\newcommand{\bb}{{\bf b}}

\newcommand{\Sb}{{\bf S}}

\newcommand{\beps}{\varepsilon} 
\newcommand{\FMS}{\mbox{\sf FMS}}

\newcommand{\Leff}{\mbox{\sf L}} 
\newcommand{\ALGP}{\mbox{\sf ALGP}}
\newcommand{\ALGD}{b\mbox{\sf ALGD}}
\newcommand{\FMSeq}[1]{\mbox{(\hyperref[#1]{F\ref*{#1}})}}
\newcommand{\PRIMAL}{\ensuremath{\mathop{\mathsf{PRIMAL}}}}
\newcommand{\DUAL}{\ensuremath{\mathop{\mathsf{DUAL}}}}
\newcommand{\MWUMa}{\textsf{MWUM}-$\alpha$\xspace}
\newcommand{\Sbb}{\Sb_\bb}
\newcommand{\fbb}{{\mathsf F}_\bb}
\newcommand{\fD}{\ensuremath{\mathfrak{D}}}
\newcommand{\fR}{\ensuremath{\mathfrak{R}}}

\newcommand{\CR}{{\sf CR}}
\newcommand{\CRE}{\CR^*}

\begin{document}
\begin{frontmatter}

\title{On the Flow-level Dynamics of a Packet-switched Network}
\runtitle{On the Flow-level Dynamics of a Packet-switched Network}

\begin{aug}
  \author{\fnms{Ciamac C.} \snm{Moallemi}
    \ead[label=e1]{ciamac@gsb.columbia.edu}
    \and
    \fnms{Devavrat} \snm{Shah}
    \thanksref{t2}
    \ead[label=e2]{devavrat@mit.edu}
  }
  \runauthor{C. C. Moallemi \& D. Shah}

  \affiliation{Columbia University and Massachusetts Institute of
    Technology}

  \thankstext{t2}{C. C. Moallemi is at the Graduate School of Business,
    Columbia University. D. Shah is with Laboratory for Information and
    Decision Systems, Massachusetts Institute of Technology. This work was
    supported in parts by NSF projects CNS 0546590 and TF 0728554. 
    Authors' email addresses: \emailhref{ciamac@gsb.columbia.edu},
    \emailhref{devavrat@mit.edu}. Initial version: November 10, 2009.}
\end{aug}

\maketitle

\begin{abstract}
  The packet is the fundamental unit of transportation in modern
  communication networks such as the Internet. Physical layer scheduling
  decisions are made at the level of packets, and packet-level models with
  exogenous arrival processes have long been employed to study network
  performance, as well as design scheduling policies that more efficiently
  utilize network resources.
  On the other hand, a user of the network is more concerned with
  end-to-end bandwidth, which is allocated through congestion
  control policies such as TCP. Utility-based flow-level models
  have played an important role in understanding congestion control protocols.
  In summary, these two classes of models have provided separate insights
  for flow-level and packet-level dynamics of a network.

  In this paper, we wish to study these two dynamics together. We propose a
  joint flow-level and packet-level stochastic model for the dynamics of a
  network, and an associated policy for congestion control and packet
  scheduling that is based on $\alpha$-weighted policies from the
  literature. We provide a fluid analysis for the model that establishes the
  throughput optimality of the proposed policy, thus validating prior insights
  based on separate packet-level and flow-level models. By analyzing a
  critically scaled fluid model under the proposed policy, we provide constant
  factor performance bounds on the delay performance and characterize the
  invariant states of the system.

\end{abstract}

\begin{keyword}
\kwd{Flow-level model}
\kwd{Packet-level model}
\kwd{Congestion control}
\kwd{Scheduling}
\kwd{Utility maximization}
\kwd{Back-pressure maximum weight}
\end{keyword}

\end{frontmatter}

\section{Introduction}

The optimal control of a modern, packet-switched data network can be
considered from two distinct vantage points. From the first point of view, the
atomic unit of the network is the \emph{packet}. In a packet-level model, the
limited resources of a network are allocated via the decisions on the
scheduling of packets. Scheduling policies for packet-based networks have been
studied across a long line of literature (e.g., \cite{TE92, Stolyar04,SW06}).
The insights from this literature have enabled the design of scheduling
policies that allow for the efficient utilization of the resources of a
network, in the sense of maximizing the throughput of packets across the
network, while minimizing the delay incurred by packets, or, equivalently, the
size of the buffers needed to queue packets in the network.

Packet-level models accurately describe the mechanics of a network at a
low level. However, they model the arrival of new packets to the network
exogenously. In reality, the arrival of new packets is also under the control
of the network designer, via rate allocation or congestion control decisions.
Moreover, while efficient utilization of network resources is a reasonable
objective, a network designer may also be concerned with the
satisfaction of end users of the network. Such objectives cannot directly be
addressed in a packet-level model.

Flow-level models (cf.\ \cite{KMT98, BM01}) provide a different point of view
by considering the network at a higher level of abstraction or, alternatively,
over a longer time horizon. In a flow-level model, the atomic unit of the
network is a \emph{flow}, or user, who wishes to transmit data from a source
to a destination. Resource allocation decisions are made via the allocation of
a transmission rate to each flow. Each flow generates utility as a function of
its rate allocation, and rate allocation decisions may be made so as to
maximize a global utility function. In this way, a network designer can
address end user concerns such as fairness.

Flow-level models typically make two simplifying assumptions. The first
assumption is that, as the number of flows evolves stochastically over time,
the rates allocated to flows are updated \emph{instantaneously}. The rate
allocation decision for a particular flow is made in a manner that requires
immediate knowledge of the demands of other flows for the limited transmission
resources along the flow's entire path. This assumption, referred to as
\emph{time-scale separation}, is based on the idea that flows arrive and
depart according to much slower processes than the mechanisms of the rate
control algorithm. The second assumption is that, once the rate allocation
decision is made, each flow can transmit data \emph{instantaneously} across
the network at its given rate. In reality, each flow generates discrete
packets, and these packets must travel through queues to traverse the network.
Moreover, the packet scheduling decisions within the network must be made in a
manner that is consistent with and can sustain the transmission rates
allocated to each flow, and the induced packet arrival process must not result
in the inefficient allocation of low level network resources.

In this paper, our goal is to develop a stochastic model that jointly captures
the packet-level and flow-level dynamics of a network, without any
assumption of time-scale separation. The contributions of this paper are as
follows: 

\begin{enumerate}
  \itemsep=0pt
\item We present a joint model where the dynamic evolution of flows and
  packets is simultaneous. In our model, it is possible to simultaneously seek
  efficient allocation of low level network resources (buffers) while
  maximizing the high-level metric of end-user utility.

\item For our network model, we propose packet scheduling and rate allocation
  policies where decisions are made
  via myopic algorithms that combine the distinct insights of prior packet-
  and flow-level models. Packets are scheduled according to a maximum weight
  policy. The rate allocation decisions are completely local
  and distributed. Further, in long term (i.e.,
  under fluid
  scaling), the rate control policy exhibits the behavior of a primal
  algorithm for an appropriate utility maximization problem.
  
\item We provide a fluid analysis of the joint packet- and flow-level
  model. This analysis allows us to establish stability of the joint model and
  the throughput optimality of our proposed control policy.

\item We establish, using a fluid model under critical loading,
  a performance bound on our control policy under the metric of
  minimizing the outstanding number of packets and flows in the network (or, in
  other words, minimizing delay). We demonstrate that, for a class of
  \emph{balanced} networks, our control policy performs to within a constant
  factor of any other control policy.

\item Under critical loading, we characterize the invariant manifold of the
  fluid model of our control policy, as well as establishing convergence to
  this manifold starting from any initial state. These results, along with the
  method of Bramson \cite{Bramson98}, lead to the characterization of
  multiplicative state space collapse under \emph{heavy traffic
    scaling}. Further, we establish that the invariant states of the fluid
  model are asymptotically optimal under a limiting control policy.

\end{enumerate}

In summary, our work provides a joint dynamic flow- and packet-level model
that captures the microscopic (packet) and macroscopic (fluid, flow) behavior
of large packet-based communications network faithfully. The performance
analysis of our rate control and scheduling algorithm suggests that the
separate insights obtained for dynamic flow-level models \cite{KMT98, BM01}
and for packet-level models \cite{TE92, Stolyar04, SW06} indeed continue to
hold in the combined model.

The balance of the paper is organized as follows. In
Section~\ref{sec:lit-review}, we survey the related literature on flow- and
packet-level models. In Section~\ref{sec:model}, we introduce our network
model. Our network control policy, which combines features of maximum weight
scheduling and utility-based rate allocation is described in
Section~\ref{sec:mwum}. A fluid model is derived in
Section~\ref{sec:fluid-model}. Stability (or, throughput optimality) of the
network control policy is established in Section~\ref{sec:stability}. The
critically scaled fluid model is described in
Section~\ref{sec:critical-loading}. In Section~\ref{sec:balanced}, we provide
performance guarantees for balanced networks. The invariant states of the
critically scaled fluid model are described in Section~\ref{sec:invariant}.
Finally, in Section~\ref{sec:conclusion}, we conclude.

\subsection{Literature Review}\label{sec:lit-review}

The literature on scheduling in packet-level networks begins with
Tassiulas and Ephremides \cite{TE92}, who proposed a class of
`maximum weight' (MW) or `back-pressure' policies. Such policies assign a
weight to every schedule, which is computed by summing the number of packets
queued at links that the schedule will serve. At each instant of time, the
schedule with the maximum weight will be selected. Tassiulas and Ephremides
\cite{TE92} establish that, in the context of multi-hop wireless networks, MW
is throughput optimal. That is, the stability region of MW contains the
stability region of any other scheduling algorithm. This work was subsequently
extended to a much broader class of queueing networks by others (e.g.,
\cite{MAW96,DP00,TB00,DL05}).

By allowing for a broader class of weight functions, the MW algorithm can be
generalized to the family of so-called MW-$\alpha$ scheduling algorithms.
These algorithms are parameterized by a scalar $\alpha\in (0,\infty)$.
MW-$\alpha$ can be shown to inherit the throughput optimality of MW
\cite{KM01,Shah01} for all values of $\alpha\in (0,\infty)$. However, it has
been observed experimentally that the average queue length (or, `delay') under
MW-$\alpha$ decreases as $\alpha\tends 0^+$ \cite{KM01}. Certain delay
properties of this class of algorithms have been subsequently established
under a heavy traffic scaling \cite{Stolyar04,DL05,SW06}.

Flow-level models have received significant recent attention in the
literature, beginning with the work of Kelly, Maulloo, and Tan \cite{KMT98}.
This work developed rate-control algorithms as decentralized solutions to a
deterministic utility maximization problem. This optimization problem seeks to
maximize the utility generated by a rate allocation, subject to capacity
constraints that define a set of feasible rates. This work was subsequently
generalized to settings where flows stochastically depart and arrive
\cite{MR00,dVLK01,BM01}, addressing the question of the stability of the
resulting control policies. Fluid and diffusion approximations of the
resulting systems have been subsequently developed
\cite{KW04,KKLW04,YY08}. Under these stochastic models, flows are assumed to
be allocated rate as per the optimal solution of the utility maximization
problem instantaneously. Essentially, this \emph{time-scale separation}
assumption captures the intuition that the dynamics of the arrivals and
departures of flows happens on a much slower time-scale than the dynamics of
rate control algorithm.

In reality, flow arrivals/departures and rate control happen on the
\emph{same} time-scale. Various authors have considered this issue, in the
context of understanding the stability of the stochastic flow level models
without the time-scale separation assumption
\cite{LSS08,ES06,Stolyar05,Neely05,SS08}. Lin, Schroff, and Srikant
\cite{LSS08} assume a stochastic model of flow arrivals and departures as well
as the operation of a primal-dual algorithm for rate allocation. However,
there are no packet dynamics present. Other work \cite{ES06,
  Stolyar05,Neely05} has assumed that rate control for each type of flow is a
function of a local Lagrange multiplier; and a separate Lagrange multiplier is
associated with each link in the network. These multipliers are updated using
a maximum weight-type policy. In this line of work, Lagrange multipliers are
interpreted as queue lengths, but there are no actual packet-level dynamics
present. Further, these models lack flow-level dynamics as well. Thus, while
overall this collection of work is closest to the results of this paper, it
stops short of offering a complete characterization of a joint flow- and
packet-level dynamic model.

Finally, we take note of recent work by Walton \cite{Walton08}, which presents
a simple but insightful model for joint flow- and packet-level dynamics. In
this model, each source generates packets by reacting to the acknowledgements
from its destination, and at each time instant, each source has at most one
packet in flight. Under a many-source scaling for a specific network topology,
it is shown that the network operates with rate allocation as per the
proportional fair criteria. This work provides important intuition about the
relationship between utility maximization and the rate allocation resulting
from the packet-level dynamics in a large network. However, it is far from
providing a comprehensive joint flow- and packet-level dynamic model as well
as efficient control policy.

\section{Network Model}\label{sec:model}

In this section, we introduce our network model. This model captures both the
flow-level and the packet-level aspects of a network, and will allow us to
study the interplay between the dynamics at these two levels. In a nutshell,
flows of various types arrive according to an exogenous process and seek to
transmit some amount of data through the network. As in the standard
congestion control algorithm, TCP, the flows generate packets at their ingress
to the network. The packets travel to their respective destinations along
links in the network, queueing in buffers at intermediate locations. As each
packet travels along its route, it is subject to physical layer constraints,
such as medium access constraints, switching constraints, or constraints due
to limited link capacity.  A flow departs once all of its packets have been
sent.

In this section, we focus on the mechanics of the network that are independent
of the network control policy. In Section~\ref{sec:mwum}, we will propose a
specific network control policy to be applied in the context of this model.

\subsection{Network  Structure}

Consider a network consisting of a set $\Vscr$ of destination nodes, a set
$\Lscr$ of links, and a set $\Fscr$ of flow types.  Each flow type is
identified by a fixed given route starting at the source link $s(f)\in\Lscr$
and ending at the destination node $d(f)\in\Vscr$. At a given time, multiple
flows of a given type exist in the network, each flow injects packets into the
network.

The network maintains buffers for packets that are in transit across the
network. At each link, there is a separate queue for the packets corresponding
to each possible destination. Let $\Escr = \Lscr \times \Vscr$ denote the set of
all such queues, with each $e = (\ell, v)$ being the queue at link $\ell$ for
final destination $v$. Traffic in each queue is transmitted to the next hop
along the route to the destination, and leaves the network when it reaches the
destination. We define the routing matrix $R\in\{0,1\}^{\Escr\times\Escr}$ by
setting $R_{ee'}\defeq 1$ if the next hop for queue $e$ is queue $e'$, and
$R_{ee'}\defeq 0$ otherwise.  Traffic for a flow of type $f$ enters the
network in the ingress queue $\iota(f)\defeq \big(s(f),d(f)\big)\in\Escr$.
Define the ingress matrix $\Gamma \in \{0,1\}^{\Escr\times \Fscr}$ by setting
$\Gamma_{ef} \defeq 1$ if $\iota(f)=e$, and $\Gamma_{ef} \defeq 0$ otherwise.
We will assume that the routes are acyclic. In this case, we can define the
matrix
\begin{equation}\label{eq:xi-def}
\Xi \defeq \big(I-R^\top\big)^{-1} = I + R^\top + (R^\top)^2 + \cdots.
\end{equation}
Under the acyclic routing assumption, $\Xi_{e'e} = 1$ if and only if a packet 
arriving at queue $e$ will subsequently eventually pass through queue $e'$.

\subsection{Dynamics: Flow-Level}\label{sec:dyn-flow}

In this section, we will describe in detail the stochastic model for dynamics
of flows in the network. The system evolves in continuous time, with $t \in
[0,\infty)$ denoting time, starting at $t = 0$.  For each flow type
$f\in\Fscr$, let $N_f(t)$ denote the number of flows of type $f$ active at
time $t$. Flows of type $f$ arrive according to an independent Poisson process
of rate $\nu_f$. Flows of type $f$ receive an \emph{aggregate} rate of service
$X_f(t) \in [0,C]$ at time $t$. Here, $C>0$ is the maximal rate of service
that can be provided to any flow type. The total rate of service $X_f(t)$ is
divided equally amongst the $N_f(t)$ flows.  As flows are serviced, packets
are generated. The evolution of packets and flows proceeds according to the
following:
\begin{itemize}
  \itemsep=0pt
\item Packets are generated by all the flows of type $f$, in aggregate, as a
  time varying Poisson process of rate $X_f(t)$ at time $t$. If
  $N_f(t)=0$, then we require that $X_f(t) = 0$.
\item When a packet is generated by a flow of type $f$, it joins the
  ingress queue $\iota(f)\in\Escr$.
\item When a packet is generated by a flow of type $f$, the flow departs from
  the network with a probability of $0<\mu_f< 1$, independent of everything
  else. 
\end{itemize}
Thus, each flow of type $f$ generates a number of packets that is distributed
according to an independent geometric random variable with mean\footnote{ The
  assumption that $\mu_f < 1$ is equivalent to requiring that $1/\mu_f > 1$,
  i.e., each flow is expected to generate more than one packet. This is
  reasonable since we require flows to arrive with at least one packet and for
  there to be some variability in the number of packets associated with a
  flow.}  $1/\mu_f$, and the flow departure process for flows of type $f$ is a
Poisson process of rate $\mu_f X_f(t)$ at time $t$. We can summarize the flow
count process $N_f(\cdot)$ by the transitions
\[
N_f(t) \rightarrow
\begin{cases}
N_f(t) + 1 & \text{at rate $\nu_f$,} \\
N_f(t) - 1 & \text{at rate $\mu_f X_f(t)$.}
\end{cases}
\]

Define the \emph{offered load} vector $\rho\in\Rp^\Fscr$ by $\rho_f \defeq
\nu_f/\mu_f$, for each flow type $f$. Without loss of generality, we will make
the following assumptions:\footnote{In what follows, inequalities between
  vectors are to be interpreted component-wise. The vector $\bzero$ (resp.,
  $\bone$) is the vector where every component is $0$ (resp., 1), and whose
  dimension should be inferred from the context.}
\begin{itemize}
  \itemsep=0pt
\item $\rho > \bzero$, i.e., we restrict attention to flows with a non-trivial
  loading.
\item $\rho < C\1$, i.e., we assume that the maximal service rate $C$ is
  sufficient for the load generated by any single flow type.
\item $\Xi \Gamma \rho > \bzero$, i.e., we restrict attention to queues with a
  non-trivial loadings.
\end{itemize}

Denote by $\An_f(t)$ the cumulative number of flows of type $f$ that have
arrived in the time interval $[0,t]$. Denote the cumulative number of packets
generated by flows of type $f$ in the time interval $[0,t]$ by $A_f(t)$.  We
suggest that the reader take note of difference between $\An_f(\cdot)$ and
$A_f(\cdot)$. Let $D_f(t)$ denote the cumulative number of flows of type $f$
that have departed in the time interval $[0,t]$. The evolution of the flow count
for flow type $f$ over time can be written as
\begin{equation}\label{eq:prim-f}
N_f(t) = N_f(0) + \An_f(t) - D_f(t). 
\end{equation}

\subsection{Dynamics: Packet-Level}\label{sec:dyn-packet}

As we have just described, flows generate packets which are injected into the
network. These packets must traverse the links of the network from source to
destination. In this section, we describe the dynamics of packets in the
network.

We assume that each queue in the network is capable of transmitting at most 1
data packet per unit time. However, the collection of queues that can
simultaneously transmit is restricted by a set of scheduling
constraints. These scheduling constraints are meant to capture any limitations
of the network due to scarce resources (e.g., limited wireless bandwidth,
limited link capacity, switching constraints, etc.).

Formally, the scheduling constraints are described by the set $\Sscr\subset
\{0,1\}^\Escr$. Under a permissible schedule $\pi \in \Sscr$, a packet will be
transmitted from a queue $e\in\Escr$ if and only if $\pi_e=1$. We assume that
$\bzero\in\Sscr$. We require that each queue $e$ be served by \emph{some}
schedule, i.e., there exists a $\pi\in\Sscr$ with
$\pi_e = 1$.Further, we assume that $\Sscr$ is \emph{monotone}: if $\sigma \in
\Sscr$ and $\sigma' \in \{0,1\}^\Escr$ is such that $\sigma'_e \leq \sigma_e$
for every queue $e$, then $\sigma' \in \Sscr$.  Finally, denote by
$\Pi\in\{0,1\}^{\Escr\times\Sscr}$ the matrix with columns consisting of the
elements of $\Sscr$.

We assume that the scheduling of packets happens at every integer time. At a
time $\tau\in\Zp$, let $\pi(\tau)\in\Sscr$ denote the scheduled queues
for the time interval $[\tau,\tau+1)$. For each queue $e$, denote by
$Q_e(\tau^-)$ the length of the queue $e$ immediately prior to the time $\tau$
(i.e., before scheduling happens). The queue length evolves, for times
$t\in[\tau,\tau+1)$ according to\footnote{$\I{\cdot}$ denotes the indicator
  function.}
\[
\begin{split}
Q_e(t)  & \defeq   Q_e(\tau^-) - \pi_e(\tau) \I{Q_e(\tau^-) > 0}
+ \sum_{f\in\Fscr} \Gamma_{ef} \bigl(A_f(t) 
- A_f(\tau^-)\bigr) 
+ \sum_{e'\in\Escr} R_{e'e} \pi_{e'}(\tau) \I{Q_{e'}(\tau^-) > 0}.
\end{split}
\]
Here, for each flow type $f$, $A_f(\tau^-)$ is the cumulative number of
packets generated by flows of type $f$ in the time interval $[0,\tau)$.  The
term $\pi_e(\tau) \I{Q_e(\tau) > 0}$ enforces an idling constraint, i.e., if
queue $e$ is scheduled but empty, no packet departs.  Note that, over a time
interval $[\tau,\tau+1)$, we assume the transmission of packets already
present in the network occurs instantly at time $\tau$, while the arrival of
new packets to the network occurs continuously throughout the entire time
interval.

Finally, let $S_\pi(\tau)$ denote the cumulative number of time slots during
which the schedule $\pi$ was employed up to and including time $\tau$. Let
$Z_e(\tau)$ denote the cumulative idling time for queue $e$ up to and including
time $\tau$. That is,
\[
Z_e(\tau) \defeq \sum_{s=0}^\tau \sum_{\pi \in \Sscr}  \pi_e 
\big(S_\pi(s) - S_\pi(s-1)\big) \I{Q_e(s) = 0}.
\]
Then, the overall queue length evolution can be written in vector form as
\begin{equation}\label{eq:prim-q}
\begin{split}
Q(\tau+1) 
& = Q(0) - \big(I-R^\top\big) \Pi S(\tau) 
+ \big(I-R^\top\big) Z(\tau) 
+ \Gamma A(\tau+1), 
\end{split}
\end{equation}
where we define the vectors
\begin{align*}
Q(t) 
& \defeq \big[Q_e(t)\big]_{e\in\Escr},
&
A(t) 
& \defeq \big[A_f(t)\big]_{f\in\Fscr}, \\
S(\tau) & \defeq \big[S_\pi(\tau)\big]_{\pi \in \Sscr},
& 
Z(\tau) & \defeq \big[Z_e(t)\big]_{e\in\Escr}.
\end{align*}

\section{MWUM Control Policy}\label{sec:mwum}

A network control policy is a rule that, at each point in time, provides two
types of decisions: (a) the rate of service provided to each flow, and (b)
the scheduling of packets subject to the physical constraints in the network.
In Section~\ref{sec:model}, we described the stochastic evolution of flows and
packets in the network, taking as given the network control policy. In this
section, we describe a control policy called the \emph{maximum weight utility
  maximization-$\alpha$} (\MWUMa) policy. \MWUMa takes as a parameter a scalar
$\alpha\in(0,\infty)\setminus\{1\}$.

The \MWUMa policy is myopic and based only on local information.
Specifically, a flow generates packets at rate that is based on the queue
length at its ingress, and the scheduling of packets is decided as a function
of the effected queue lengths. 

At the flow-level, rate allocation decisions are made according to a per flow
utility maximization problem. Each flow chooses a rate so as to myopically
maximize its utility as a function of rate consumption, subject to a linear
penalty (or, `price') for consuming limited network resources. As in the case
of $\alpha$-fair rate allocation, the utility function is assumed to have a
constant relative risk aversion of $\alpha$. The price charged is a function
of the number of packets queued at the ingress queue associated with the flow,
raised to the $\alpha$ power.

At the packet-level, packets are scheduled according to a maximum
weight-$\alpha$ scheduling algorithm. In particular, each queue is assigned a
weight equal to the number of queued packets to the $\alpha$ power, and a
schedule is picked which maximizes the total weight of all scheduled queues.


\subsection{Control: Rate Allocation}

The first control decision we shall consider is that of \emph{rate
  allocation}, or, the determination of the aggregate rate of service
$X_f(t)$, at time $t$, for flows of type $f$. We will assume our network is
governed by a variant of an $\alpha$-fair rate allocation policy. This is as
follows:

Assume that each flow of type $f$ is allocated a rate $Y_f(t) \geq 0$ at time
$t$ by maximizing a (per flow) utility function that depends on the allocated
rate, subject to a linear penalty (or, cost) for consuming resources from the
limited capacity of the network.  In particular, we will assume a utility
function given a rate allocation of $y\geq 0$ to an individual flow of type
$f$ of the form
\[
V_f(y) \defeq
\frac{y^{1-\alpha}}{1-\alpha},
\]
for some $\alpha\in(0,\infty)\setminus\{1\}$.  This utility function is
popularly known as $\alpha$-fair in the congestion control literature
\cite{MW00}, and has a constant relative risk aversion of $\alpha$.  The
individual flow will be assigned capacity according to
\[
Y_f(t) \in \argmax_{y \geq 0}\ V_f(y) - Q^\alpha_{\iota(f)}(t) y.
\]
Here, $Q^\alpha_{i(f)}(t)$ represents a `price' or `congestion
signal'. Intuitively, a flow reacts to the congestion in the network (or lack
thereof) through the length of the ingress or `first-hop' queue $\iota(f)$.
Then, if $N_f(t)>0$, the aggregate rate $X_f(t)$ allocated to all flows of
type $f$ at time $t$ is determined according to
\[
\begin{split}
X_f(t) & = N_f(t) Y_f(t) 
= 
\argmax_{x \geq 0}\ 
\frac{x^{1-\alpha} N_f^\alpha(t)}{1-\alpha} - 
Q^\alpha_{\iota(f)}(t)\, x.
\end{split}
\]

If $N_f(t)=0$, we require that $X_f(t)=0$. Further, we will constrain the
overall rate allocated to flows of type $f$ by the constant $C$. Thus, rate
allocation is determined by the equation
\[
X_f(t) = 
\begin{cases}
\displaystyle
\argmax_{x\in [0,C]}\ 
\frac{x^{1-\alpha} N_f^\alpha(t)}{1-\alpha} - 
Q^\alpha_{\iota(f)}(t)\, x
& \text{if $N_f(t)>0$,}
\\
0 & \text{otherwise.}
\end{cases}
\]
Given the strictly concave nature of the objective in this
optimization program, it is clear that the maximizer is unique and
$X_f(t)$ is well-defined.

Denote by $\bar{X}_f(t)$ the cumulative rate allocation to flows of type $f$
in the time interval $[0,t]$, i.e.,
\[
\bar{X}_f(t) \defeq \int_{0}^t X_f(s)\,ds. 
\]
$\bar{X}_f(\cdot)$ is Lipschitz continuous and differentiable,
since $X_f(\cdot)$ is always bounded by $C$.  

\subsection{Control: Scheduling}

The second control decision that must be specified is that of scheduling. We
will assume the following variation of the `maximum weight' or
`back-pressure' policies introduced by Tassiulas and Ephremides \cite{TE92}.

At the beginning of each discrete-time slot $\tau\in\Zp$, a schedule
$\pi(\tau)\in\Sscr$ is chosen according to the optimization problem
\begin{equation}\label{eq:pi-control}
\begin{split}
\pi(\tau)  & \in
\argmax_{\pi\in\Sscr}\ 
\sum_{e\in\Escr} \pi_e \left[
Q_e^\alpha(\tau^-) - \sum_{e'\in\Escr} R_{ee'} Q^\alpha_{e'}(\tau^-)
\right]
\\
& =
\argmax_{\pi\in\Sscr}\ 
\pi^\top (I - R) Q^\alpha(\tau^-),
\end{split}
\end{equation}
where 
$Q^\alpha(\tau^-) \defeq \big[Q^\alpha_e(\tau^-)\big]_{e\in\Escr}$
is a vector of component-wise powers of queue lengths, immediately prior to
time $\tau$. In other words, the schedule $\pi(\tau)$ is chosen so as to
maximize the summation of weights of queues served, where weight of a queue $e
\in \Escr$ is given by $\big[(I-R)Q^\alpha(\tau^-)\big]_e$. Given that $\Sscr$
is monotone, there exists a $\pi \in \Sscr$ that maximizes this weight and is
such that $\pi_e = 0$ if $Q_e(\tau^-) = 0$.  We will restrict our attention to
such schedules only.  From this, it follows that the objective value of the
optimization program in \eqref{eq:pi-control} is always non-negative.

By the discussion above, it is clear that the following invariants are
satisfied:
\begin{enumerate}
  \itemsep=0pt
\item For any schedule $\pi$ and time $\tau$, $S_\pi(\tau) = S_\pi(\tau-1)$ if
  \[
  \pi^\top (I - R) Q^\alpha(\tau^-) < \sigma^\top (I - R) Q^\alpha(\tau^-),
  \]
  for some $\sigma \in \Sscr$.

\item For any queue $e$ and time $\tau$, $Z_e(\tau) = 0$.  In other words,
  there is no idling.
\end{enumerate}

\section{Fluid Model}\label{sec:fluid-model}


In this section, we introduce the fluid model of the our system.  As we shall
see, the allocation of rate to flows in the fluid model resembles rate
allocation of `flow-level' models that has been popular in the literature
\cite{MR00,BM01}.  In that sense, our model on original time-scale operates at
a packet-level granularity and on the fluid-scale operates at a flow-level
granularity.

\subsection{Fluid Scaling and Fluid Model  Equations} 

In order to introduce the fluid model of our network, we will consider the
scaled version of the system.  To this end, denote the overall system state at
a time $t\geq 0$ by
\[
\Zscr(t) \defeq \Big(Q(t),Z(\lfloor t \rfloor), 
N(t),S(\lfloor t \rfloor),\bar{X}(t), \An(t), D(t), A(t)\Big).
\]
Here, the components of the state $\Zscr(t)$ are the primitives introduced in
Sections~\ref{sec:dyn-flow} and \ref{sec:dyn-packet}.  That is, at times
$t\in\Rp$, we have,
\[
\begin{array}{l@{\ }l@{\qquad}l}
Q(t) & \defeq \big[Q_e(t)\big]_{e\in\Escr},
& \text{where
$Q_e(t)$ is the length of queue $e$},
\\[0.05in]
N(t) & \defeq \big[N_f(t)\big]_{f\in\Fscr},
& \text{where $N_f(t)$ 
is the number of flows of type $f$},
\\[0.05in]
\bar{X}(t) &
\defeq \big[\bar{X}_f(t)\big]_{f\in\Fscr},
&
\text{where
$\bar{X}_f(t)$ is the cumulative rate allocated to flow type
  $f$},
\\[0.05in]
\An(t) & \defeq \big[\An_f(t)\big]_{f\in \Fscr},
& \text{where
$\An_f(t)$  is the cumulative arrival count of flow type
$f$},
\\[0.05in]
D(t) & \defeq \big[D_f(t)\big]_{f\in\Fscr},
& \text{where $D_f(t)$ is
the cumulative departure count of flow type $f$},
\\[0.05in]
A(t) & \defeq \big[A_f(t)\big]_{f\in \Fscr},
& \text{where
$A_f(t)$  is the cumulative packet arrival count of flow type $f$},
\end{array}
\]
and, at times $\tau\in\Zp$, we have
\[
\begin{array}{l@{\ }l@{\qquad}l}
Z(\tau) & \defeq \big[Z_e(\tau)\big]_{e\in\Escr},
&
\text{where
$Z_e(\tau)$ is the cumulative idleness for queue $e$},
\\[0.05in]
S(\tau) & \defeq \big[S_\pi(\tau)\big]_{\pi\in\Sscr},
&
\text{where
$S_\pi(\tau)$ is the cumulative time schedule
$\pi$ is employed}.
\end{array}
\]

Given a scaling parameter $r\in \R$, $r\geq 1$, define the scaled system state
as
\begin{equation}\label{eq:scaled-state}
\begin{split}
\Zscr^{(r)}(t)  \defeq
\Big( & Q^{(r)}(t),Z^{(r)}(t),
N^{(r)}(t),S^{(r)}(t), 
\bar{X}^{(r)}(t),
\An^{(r)}(t), D^{(r)}(t), A^{(r)}(t)\Big).
\end{split}
\end{equation}
Here, the scaled components are defined as
\[
\begin{array}{l@{\,}l@{\quad}l@{\,}l}
Q^{(r)}(t) & \defeq r^{-1} Q(rt), 
&
N^{(r)}(t) & \defeq r^{-1} N(rt), \\[4pt]
\bar{X}^{(r)}(t) & 
\defeq r^{-1} \bar{X}(rt), 
&
\An^{(r)}(t) & \defeq r^{-1} \An(rt), 
\\[4pt]
D^{(r)}(t) & \defeq r^{-1} D(rt), 
&
A^{(r)}(t) & \defeq r^{-1} A(rt),
\\[4pt]
Z^{(r)}(t) 
&
\multicolumn{3}{l}{
\!\!\!
\defeq r^{-1} 
\big[(rt - \lfloor rt \rfloor)Z(\lceil rt \rceil) 
+ (\lceil rt \rceil- rt)Z(\lfloor rt \rfloor)\big], 
}
\\[4pt]
S^{(r)}(t) 
&
\multicolumn{3}{l}{
\!\!\!
\defeq r^{-1} \big[(rt - \lfloor rt \rfloor)S(\lceil rt \rceil)
+ (\lceil rt \rceil- rt)S(\lfloor rt \rfloor)\big].
}
\end{array}
\]
In the above, the components $Z(\cdot)$ and $S(\cdot)$ are linearly
interpolated for technical convenience only.

Our interest is in understanding the behavior of $\Zscr^{(r)}(\cdot)$ as
$r\to\infty$. Roughly speaking, in this limiting system the trajectories will
satisfy certain deterministic equations, called fluid model
equations. Solutions to these equations, which are defined below, will be
denoted as fluid model solutions. The formal result is stated in
Theorem~\ref{thm:fms}.

\begin{definition}[\textbf{Fluid Model Solution}]\label{def:fms}
  \ 
  Given fixed initial conditions $q(0)\in\Rp^\Escr$ and $n(0)\in\Rp^\Fscr$,
  for every time horizon $T > 0$, let $\FMS(T)$ denote the set of all
  trajectories
\begin{multline}\label{eq:fms-path}
  \zscr(t) \defeq \big(q(t),z(t),n(t),s(t),\bar{x}(t), \an(t), d(t), a(t)\big)
  \\
  \in \mathfrak{Z} \defeq
  \Rp^\Escr\times\Rp^\Escr\times\Rp^\Fscr\times\Rp^\Sscr\times\Rp^\Fscr
  \times\Rp^\Fscr\times\Rp^\Fscr\times\Rp^\Fscr,
\end{multline}
over the time interval $[0,T]$ such that:
\newcounter{FMS-enum}
\begin{enumerate}[({F}1)]
  \itemsep=0pt
\item\label{enum:one} All components of $\zscr(t)$ are uniformly Lipschitz
  continuous and thus differentiable for almost every $t\in (0,T)$. Such
  values of $t$ are known as regular points.

\item\label{enum:two} For all $t\in [0,T]$,
$n(t) = n(0) + \an(t) - d(t)$.

\item\label{enum:three} For all $t \in [0,T]$, 
$\an(t) = \nu t$. 

\item\label{enum:four} For all $t\in [0,T]$,
$d(t) = \diag(\mu)\bar{x}(t)$.

\item\label{enum:five} For all $t\in [0,T]$, 
\[
q(t) = q(0) - \big(I - R^\top\big) \Pi s(t) + \big(I-R^\top\big) z(t)
+ \Gamma a(t). 
\]

\item\label{enum:six} For all $t\in [0,T]$, 
$a(t) = \bar{x}(t)$.

\item\label{enum:seven} For all $t\in [0,T]$,
$\bone^\top s(t) =  t$.

\item\label{enum:eight} Each component of $z(\cdot)$, $s(\cdot)$, and
  $\bar{x}(\cdot)$ is non-decreasing.
\setcounter{FMS-enum}{\value{enumi}}
\end{enumerate}

In addition, define the set $\FMS^\alpha(T)$ to be the subset of trajectories
in $\FMS(T)$ that also satisfy:
\begin{enumerate}[({F}1)]
  \itemsep=0pt
\setcounter{enumi}{\value{FMS-enum}}
\item\label{enum:flow} If $t\in(0,T)$ is a regular point, then for all
  $f\in\Fscr$,
{\small \[
x_f(t) =
\begin{cases}
\displaystyle
\argmax_{x\in[0,C]}\ 
\frac{x^{1-\alpha} n_f^\alpha(t)}{1-\alpha} - 
q^\alpha_{\iota(f)}(t)\, x
&
\text{if $n_f(t)>0$,}
\\
\nu_f/\mu_f ~(=\rho_f)
& \text{otherwise,}
\end{cases}
\] } where\footnote{We use the notation $\dot{\theta}(t)$ to denote
$\frac{d}{dt}\theta(t)$ for $\theta\colon [0,T] \to \R$.} $x_f(t) \defeq
\dot{\bar{x}}_f(t)$.

\item\label{enum:packet} If $t\in(0,T)$ is a regular point, then for all $\pi\in\Sscr$,
$\dot{s}_\pi(t) = 0$,
if 
\[
\pi^\top (I - R) q^\alpha(t) 
< \max_{\sigma\in\Sscr}\ \sigma^\top (I - R) q^\alpha(t).
\]
\item\label{enum:nzero} If $t\in(0,T)$ is a regular point and $n_f(t)=0$ for
  some $f\in\Fscr$, then $q_{\iota(f)}(t)=0$.
\item\label{enum:idling} For all $t \in [0,T]$, 
$z(t) = \bzero$.
\end{enumerate}
\end{definition}

Note that \FMSeq{enum:one}--\FMSeq{enum:eight} correspond to fluid model
equations that must be satisfied under \emph{any} scheduling policy, and,
hence, are \emph{algorithm independent} fluid model equations. On the other
hand, \FMSeq{enum:flow}--\FMSeq{enum:idling} are particular to networks
controlled under the \MWUMa policy. \FMSeq{enum:flow} captures the long-term
effect of the rate allocation mechanism through the $\alpha$-fair utility
maximization based policy. Indeed, in a static resource allocation model,
\FMSeq{enum:flow} can be thought of as the \emph{primal} update in an
algorithm that seeks to allocate rates to maximize the net $\alpha$-fair
utility of the flows subject to capacity constraints. \FMSeq{enum:packet}
captures the effect of short-term packet-level behavior induced by the
scheduling algorithm. Specifically, the characteristics of the MW-$\alpha$
packet scheduling algorithm are captured by this equation.

\subsection{Formal Statement}\label{sec:formal-fms}

We wish to establish fluid model solutions as limit of the scaled system state
process $\Zscr^{(r)}(\cdot)$ as $r\to \infty$.  To this end, fix a time
horizon $T>0$. Let $\bD[0,T]$ denote the space of all functions from $[0,T]$
to $\mathfrak{Z}$, as defined in \eqref{eq:fms-path}, that are right
continuous with left limits (RCLL). We will denote the Skorohod metric on this
space by $\bd(\cdot,\cdot)$ (see Appendix~\ref{ap:fluid} for details). Given a
fixed scaling parameter $r$, consider the scaled system dynamics over interval
$[0,T]$. Each sample path $\Zscr^{(r)}(\cdot)$ of the system state is RCLL,
and hence is contained in the space $\bD[0,T]$.

The following theorem formally establishes the convergence of the scaled
system process to a fluid model solution of the form specified in
Definition~\ref{def:fms}. 
\begin{theorem}\label{thm:fms} 
  Given a fixed time horizon $T > 0$, consider a collection of scaled system
  state processes $\{\Zscr^{(r)}(\cdot)\ :\ r\geq 1\} \subset \bD[0,T]$ 
  under an arbitrary control policy. Suppose the initial conditions
\begin{equation}\label{eq:init}
 \lim_{r\to\infty} Q^{(r)}(0) = q(0),\quad
 \lim_{r\to\infty} N^{(r)}(0) = n(0),
\end{equation}
are satisfied with probability $1$. Then, for any $\beps>0$,
\[
 \liminf_{r\to \infty}\  
 \PR\Bigl(\Zscr^{(r)}(\cdot) \in \FMS_\beps(T)\Bigr) = 1.
\]
Here, $\FMS_\beps(T)$ is an $\beps$-flattening of the set $\FMS(T)$ of fluid
model solutions, i.e.,
\[
 \FMS_\beps(T) \defeq\{ \bdx \in \bD[0,T]\ :\ 
 \bd(\bdx, \by) < \beps,\ 
 \by \in \FMS(T) 
 \}.
\]
Additionally, under the \MWUMa control policy, we have that
\[
 \liminf_{r\to \infty}\ 
 \PR\Bigl(\Zscr^{(r)}(\cdot) \in \FMS_\beps^\alpha(T)\Bigr) = 1.
\]
Here, $\FMS_\beps^\alpha(T)$ is an $\beps$-flattening of the set
$\FMS^\alpha(T)$ of \MWUMa fluid model solutions, i.e.,
\[
 \FMS_\beps^\alpha(T) \defeq\{ \bdx \in \bD[0,T]\ :\ 
 \bd(\bdx, \by) < \beps,\ 
 \by \in \FMS^\alpha(T) 
 \}.
\]
\end{theorem} 
Theorem~\ref{thm:fms} can be established by following a somewhat standard
sequence of arguments (cf.\ \cite{Bramson98,SW09,KW04}).  First, the
collection of measures corresponding to the collection of random processes
$\{\Zscr^{(r)}(\cdot)\ :\ r \geq 1\}$ is shown to be tight. This establishes
that limit points must exist. Next, it is established that each limit point
must satisfy the conditions of a fluid model solution with probability $1$.
The tightness argument uses concentration properties of Poisson process along
with the Lipschitz property of queue length process. A detailed argument is
required to establish that, with probability $1$, the conditions of fluid
model solution are satisfied. The complete proof of Theorem~\ref{thm:fms} is
presented in Appendix~\ref{ap:fluid}.


\section{System Stability}\label{sec:stability}

In this section, we characterize the stability of a network under the \MWUMa
policy.  In particular, we shall see that the network evolves according to a
Markov process is positive recurrent as long as the system is
\emph{underloaded}. In other words, the system is maximally stable.  In order
to construct the stability region under the \MWUMa policy, first define the
set of \emph{packet arrival rates} by
\begin{equation}\label{eq:Lambda}
\Lambda \defeq 
\left\{ \lambda \in \Rp^{\Escr}\ :\  \exists\ s\in\Rp^\Sscr \text{ with }
\lambda \leq \Pi s,\ \1^\top s \leq 1 \right\}.
\end{equation}
Imagine that the network has no packet arrivals from flows, but instead has
packets arriving according to exogenous processes. Suppose that
$\lambda\in\Rp^\Escr$ is the vector of exogenous arrival rates, so that
packets arrived to each queue $e$ at rate $\lambda_e$. Then, it is not
difficult to see that the network would not be stable under \emph{any}
scheduling policy if $\lambda \notin \Lambda$. This is because there is at
least one queue in the network that is loaded beyond its service capacity.
Hence, the set $\Lambda$ represents the raw \emph{scheduling capacity} of the
network.

The set $\Lambda$ can alternatively be described as follows: Given a vector $\lambda\in\Rp^\Escr$, consider the
linear program
\[
\begin{array}{l@{\ }lll}
\PRIMAL(\lambda) \defeq
&
\minimize_s & \1^\top s \\
& \subjectto 
& \lambda \leq \Pi s,
\\[2pt]
& &
s \in \Rp^\Sscr.
\end{array}
\]
Clearly $\lambda\in\Lambda$ if and only if $\PRIMAL(\lambda) \leq 1$. The
quantity $\PRIMAL(\lambda)$ is called the \emph{effective load} of a system
with exogenous packet arrivals at rate $\lambda$.

Now, in our model, packets arrive to the network not through an exogenous
process, but rather, they are generated by flows. As discussed in
Section~\ref{sec:dyn-flow}, each flow type $f \in \Fscr$ generates packets at
according to an offered load of $\rho_f$. The generated packets are injected
into the network according to the ingress matrix $\Gamma$, and subsequently
travel through the network along pre-determined paths specified by the routing
matrix $R$. Let $\lambda \in \Rp^{\Escr}$ be the vector of \emph{implied
  loads} on the scheduling network due to the packets generated by flows. It
seems reasonable to relate $\lambda$ and the vector $\rho \in \Rp^{\Fscr}$ of
offered loads according to $\lambda = \Gamma \rho + R^\top \lambda$.
Equivalently, we define $\lambda \defeq \Xi \Gamma \rho$, where $\Xi$ is from
\eqref{eq:xi-def}. We define the \emph{effective load} $\Leff(\rho)$ of our
network by $\Leff(\rho) \defeq \PRIMAL(\Xi \Gamma \rho)$.

Given the above discussion, it seems natural to suspect that the network's
scheduling capacity allows it to operate effectively as long as $\Leff(\rho)
\leq 1$. This motivates the following definition:

\begin{definition}[\textbf{Admissibility}] A vector $\rho\in \Rp^{\Fscr}$ of
  offered loads \textbf{admissible} if $\Leff(\rho) \leq 1$. Similarly, $\rho$
  is \textbf{strictly admissible} if $\Leff(\rho) < 1$. Finally, $\rho$ is
  \textbf{critically admissible} if $\Leff(\rho) = 1$.
\end{definition}

We will establish system stability, or, more formally, positive recurrence,
when offered load is \emph{strictly admissible}. To this end, recall that the
system is completely described by the $\Zscr(\cdot)$ process. Under the \MWUMa
policy, the evolution of all the components of $\Zscr(\cdot)$ is entirely
determined by $\big(N(\cdot), Q(\cdot)\big)$. Further, the changes in
$\big(N(\cdot),Q(\cdot)\big)$ occur at times specified by the arrivals of a
(time-varying) Poisson process. Therefore, tuple $\big(N(\cdot),
Q(\cdot)\big)$ forms the state space of a continuous-time Markov chain. The
following is the main result of this section:

\begin{theorem}\label{thm:stability}
  Consider a network system with strictly admissible $\rho$ operating under
  the \MWUMa policy. Then, the Markov chain $\big(N(\cdot),
  Q(\cdot)\big)$ is positive recurrent.
\end{theorem}

It is worth noting that if $\Leff(\rho) > 1$, then at least one of the queues
in the network must be, on average, loaded beyond its capacity.  Hence the
network Markov process can not be positive recurrent or stable.

The proof of Theorem~\ref{thm:stability}, provided next in
Section~\ref{sec:stability-proofs}, uses a fluid model approach.  Dai
\cite{Dai95} pioneered such an approach for a class of queueing
networks. However, this result does not apply to the present
setting, and a specialized analysis is needed. Conceptually, the fluid model
approach involves two steps: (1) derive strong stability of fluid model, and
(2) use strong stability to establish the positive recurrence of the original
Markov chain. To this end, define Lyapunov function $L_\alpha$
over the vector of flow counts $n=[n_f]_{f\in\Fscr} \in \Rp^{\Fscr}$ and
the vector of queue lengths $q=[q_e]_{e\in\Escr} \in \Rp^{\Escr}$ by
\begin{equation}\label{eq:lyapunov}
L_\alpha(n, q) \defeq
 \sum_{f \in \Fscr} \frac{n_f^{1+\alpha}}{\mu_f \rho_f^\alpha} + \sum_{e\in \Escr} q_e^{1+\alpha}.
\end{equation}
The following lemma, whose proof is found in Section~\ref{sec:fluid-drain},
provides a central argument towards establishing the strong stability of the
fluid model. This lemma will also be of use later in establishing the
characterization and attractiveness of invariant manifold under critical
loading.
\begin{lemma}\label{lem:fluid-drain}
  Let $\big(n(\cdot),q(\cdot)\big)$ be, respectively, the flow count process
  and the queue length process of a fluid model solution in the set
  $\FMS^\alpha(T)$. If $\Leff(\rho) \leq 1$, then for every regular point
  $t\in (0,T)$,
  \[
  \frac{d}{dt} L_\alpha\big(n(t), q(t)\big) \leq 0.
  \]

  Suppose further that $\Leff(\rho) < 1$. Then, there exist $\delta_* > 0$ and
  $T_* > 0$ such that, for all $T > T_*$, if the initial conditions $\big(n(0),
  q(0)\big)$ satisfy
  \[
  L_\alpha\big(n(0), q(0)\big) = 1,
  \]
  then
  \[ 
  L_\alpha\big(n(t), q(t)\big) \leq 1-\delta_*,
  \quad\text{for all $t \in [T_*, T]$.}
  \]
\end{lemma}

\subsection{Proof of Theorem \ref{thm:stability}}\label{sec:stability-proofs}

The following lemma provides a sufficient condition for positive recurrence:
\begin{lemma}\label{lem:pos-rec}\cite[Theorem~8.13]{Robert03} Let $\Xscr(\cdot)$
  be an irreducible, aperiodic jump Markov process on a countable state space
  $\mathfrak{X}$. Suppose there exists a function $V\colon \mathfrak{X} \to
  \Rp$, constants $A$ and $\beps>0$, and an integrable stopping time $\tau>0$
  such that, for all $x \in \mathfrak{X}$ with $V(x) > A$,
  \[
  \E\left[V(\Xscr(\tau))\ |\ \Xscr(0) = x\right] \leq 
  V(x) - \beps\E[\tau\ |\ \Xscr(0)=x]. 
  \] 
  If the set $\{x \in \mathfrak{X}\ :\ V(x) \leq A\}$ is finite and
  $\E[V(\Xscr(1))\ |\ \Xscr(0) = x] < \infty$ for all $x \in \mathfrak{X}$,
  then the process $X(\cdot)$ is positive recurrent and ergodic.
\end{lemma}

Our proof of Theorem~\ref{thm:stability} relies on establishing the sufficient
condition for positive recurrence given by Lemma~\ref{lem:pos-rec}, using the
stability of the fluid model (Lemma~\ref{lem:fluid-drain}).  To this end, note
that under the \MWUMa policy, $\Xscr(t) \defeq \big(N(t), Q(t)\big) \in
\Z_+^{\Fscr} \times \Z_+^{\Escr}$ is a jump Markov process. We shall use
`normed' version of the Lyapunov function $L_\alpha$, defined as
\[
\ell(n, q) \defeq \left(L_\alpha(n,q)\right)^{\frac{1}{1+\alpha}},
\]
for all $(n,q)\in \R^{\Fscr}_+ \times \R^{\Escr}_+$. The role of $V$ in
Lemma~\ref{lem:pos-rec} will be played by $\ell$.  
Lemma~\ref{lem:norm} implies that there exists constants $0 < C_1 \leq
C_2 < \infty$ so that, for all $(n,q)$,
\begin{equation}\label{eq:pr1}
  C_1 \|(n,q)\|_\infty \leq \ell(n,q) \leq C_2 \|(n, q)\|_\infty.
\end{equation}
Further, for any $\kappa > 0$ and $(n, q)$, we have that $\ell(\kappa n,
\kappa q) = \kappa \ell(n, q)$.

Now, consider any sequence of initial states 
\[
\left\{ 
  x^{k} \defeq 
  \big(N^{k}(0), Q^{k}(0)\big) \in \Z_+^{\Fscr}
  \times \Z_+^{\Escr}\ :\ k\in\N 
\right\},
\] 
such that $\|x^{k}\|_\infty \to \infty$ as $k\to\infty$. For each $k$,
consider a system starting at the initial state $x^{k}$. Denote the state of
the $k$th system at time $t \geq 0$ by
\[
\Xscr^k(t) 
\defeq 
\big(N^k(t), Q^{k}(t)\big).
\]

For each $k\in\N$, define the scaling factor $r_k\defeq
\ell\big(x^{(k)}\big)$, and notice that $r_k\to \infty$.  Fix a time horizon
$T > 0$, and for the $k$th system, consider the \emph{scaled} state process,
defined for $t\in [0,T]$ as
\[
\Xscr^{(r_k)}(t) 
\defeq 
\frac{1}{r_k}
\Xscr^k(r_k t)
=
\frac{1}{r_k}
\big(N^{k}(r_k t), Q^{k}(r_k t)\big), 
\]
The descriptor of the scaled system is given, for $t \in
[0,T]$, by $\Zscr^{(r_k)}(t)$ from \eqref{eq:scaled-state}.  Let $\mu^{(r_k)}$
be its distribution on $\bD[0,T]$.

From \eqref{eq:pr1}, we have that, for any $k$,
\begin{equation}\label{eq:pr2}
  \|\Xscr^{(r_k)}(0)\|_\infty = \frac{\|\Xscr^{k}(0)\|_\infty}{r_k} \leq 1/C_1.
\end{equation}
Since the set of scaled initial conditions is compact, there must exist a
limit point and a convergence subsequence. Along this subsequence, by the
analysis of Theorem~\ref{thm:fms} (in particular, Lemma~\ref{lem:mu-tight})
the measures $\{ \mu^{(r_k)} \}$ are tight, and therefore, there exists a
measure $\mu^{(\infty)}$ that is a limit point. By restricting to a further
subsequence, we can assume, without loss of generality, that $\mu^{(r_k)}
\Rightarrow \mu^{(\infty)}$ as $k\to \infty$. That is,
\[
\big(\Zscr^{(r_k)}(t)\big)_{t\in [0,T]} \Rightarrow 
\big(\zscr(t)\big)_{t\in [0,T]},
\quad \text{as $k\to\infty$,}
\]
with $\zscr(\cdot)\in\FMS^\alpha(T)$ satisfying the fluid model equations.

Given a fluid model solution $\zscr(\cdot)$ of the form \eqref{eq:fms-path},
denote by $\big(n(\cdot),q(\cdot)\big)$ the flow count and queue length
components.  Note that $\ell\big(\Xscr^{(r_k)}(0)\big)=1$, for all $k$.  By
the continuity of $\ell$, we have that $\ell\big(n(0),q(0)) = 1$; that is
$L_\alpha\big(n(0),q(0)\big) = 1$.  Then, from Lemma~\ref{lem:fluid-drain},
there exist $\delta_* > 0$ and $T_* > 0$ so that, for sufficiently large
$T$,
\begin{equation}\label{eq:pr3a}
  \ell\big(n(t),q(t)\big) \leq 1-\delta_*,
  \quad\text{for all $t \in [T_*, T]$.}
\end{equation}

Define the functional $F : \bD[0,T] \to \R_+$ by
\[
F\left(\big(\zscr(t)\big)_{t\in [0,T]}\right) \defeq
\frac{1}{T-T_*} \int_{T_*}^{T} \ell\big(n(t), q(t)\big)\, dt.
\]
Since $F$ is a continuous, it follows that  
\begin{equation}
  \label{eq:pr4}
  F\left(\big(\Zscr^{(r_k)}(t)\big)_{t \in [0,T]}\right) 
  \Rightarrow 
  F\left(\big(\zscr(t)\big)_{t\in [0,T]}\right),
  \quad \text{as $k\to\infty$.}
\end{equation}

Further, from \eqref{eq:pr2}, and using the fact that the arrival processes
are Poisson and the boundedness of the rate allocation policies, it follows
that for all $t \in [0,T]$ and all $k\in\N$,
\begin{equation}
  \label{eq:pr7}
  \E\left[\left\|\big(N^{(r_k)}(t),Q^{(r_k)}(t))\right\|_\infty\right] 
  \leq C_3 T,
\end{equation}
for some constant $C_3 > 0$.  From this, the uniform integrability of
$F\big(\Zscr^{(r_k)}(\cdot)\big)$ follows. Subsequently, from
\eqref{eq:pr3a} and \eqref{eq:pr4}, it follows that
\[
\lim_{k\to\infty} 
\E\left[
  F\left(\big(\Zscr^{(r_k)}(t)\big)_{t\in[0,T]}\right)
\right] 
\leq  1-\delta_*.
\]

Equivalently, in terms of the \emph{unscaled} 
state process $\Xscr(\cdot)$, we have that
\begin{equation}\label{eq:pr9}
  \lim_{k\to\infty} 
  \frac{1}{\ell(x^k)}
  \E\left[
    \left.
      \frac{1}{T-T_*} \int_{T_*}^{T} \ell\left(\Xscr\left(\ell\big(x^k\big) t
        \right)\right)\, 
      dt\ 
    \right|\ 
    \Xscr(0) = x^k
  \right]
  \leq 1 - \delta_*.
\end{equation}
To complete the proof of Theorem~\ref{thm:stability}, define $U$ to be a
random variable that is uniformly distributed over $[T_*,T]$. Define the
stopping time $\tau \defeq \ell\big(\Xscr(0)\big) U$.  Note that
\[
\E\left[\tau\ |\ \Xscr(0)=x\right] = \tfrac{1}{2} \ell(x) (T+T_*).
\]
Then, \eqref{eq:pr9} implies that, for all initial states
$x\in\Z^\Fscr_+\times\Z^\Escr_+$ with $\ell(x)$ sufficiently large,
\[
\E\left[
  \left.
    \ell\big(\Xscr(\tau)\big)\ 
  \right|\ 
  \Xscr(0) = x
\right]
\leq
\ell(x) - \beps \E\left[\tau\ |\ \Xscr(0)=x\right],
\]
where the constant $\beps$ is chosen so that $0 < \beps < 2 \delta_*/(T + T_*)$.
This satisfies the conditions of Lemma~\ref{lem:pos-rec}, and hence completes
the proof of Theorem~\ref{thm:stability}.

\subsection{Proof of Lemma~\ref{lem:fluid-drain}}\label{sec:fluid-drain}

Suppose $t \in (0,T)$ is a regular point. We will start by establishing the
first part of Lemma~\ref{lem:fluid-drain}: if $\Leff(\rho) \leq 1$, then $\dt
L_\alpha\big(n(t),q(t)\big) \leq 0$. To this end, note that
\begin{equation}\label{eq:dtL}
\dt 
L_\alpha\big(n(t), q(t)\big)
=
(1+\alpha)
\Bigg(
\underbrace{
  \sum_{f \in \Fscr} 
  \frac{n_f^{\alpha}(t)}{\mu_f \rho_f^\alpha}
  \dot{n}_f(t) 
}_{\mbox{\normalsize $\Delta_n$}} 
+ 
\underbrace{
  \sum_{e\in \Escr} 
  q_e^{\alpha}(t) \dot{q}_e(t)
}_{\mbox{\normalsize $\Delta_q$}} 
\Bigg).
\end{equation}
We consider the terms $\Delta_n$ and $\Delta_q$ separately.

First, consider the term $\Delta_n$. For each flow type $f$, we wish to show that
\begin{equation}\label{eq:L-A}
  \frac{n_f^{\alpha}(t)}{\mu_f \rho_f^\alpha}
  \dot{n}_f(t) 
  \leq
  q_{\iota(f)}^\alpha(t) \big(\rho_f - x_f(t)\big).
\end{equation}
By \FMSeq{enum:two}--\FMSeq{enum:four}, $\dot{n}_f(t) = \nu_f - \mu_f x_f(t)$.
There are two cases.  If $n_f(t)=0$, then by \FMSeq{enum:flow}, we have that
$x_f(t)=\rho_f$, thus both sides of \eqref{eq:L-A} are $0$. If $n_f(t) > 0$,
then
\begin{equation}\label{eq:L-A-1}
\begin{split}
  \frac{n_f^{\alpha}(t)}{\mu_f \rho_f^\alpha}
  \dot{n}_f(t) 
  =
  \frac{n_f^{\alpha}(t)}{\rho_f^\alpha}
  \big(\rho_f - x_f(t)\big)
   \leq
   n_f^\alpha(t) \frac{\rho_f^{1-\alpha}}{1-\alpha}
   -
   n_f^\alpha(t) \frac{x_f^{1-\alpha}(t)}{1-\alpha}.
\end{split}
\end{equation}
Here, the inequality follows from the fact that the function $g(z)\defeq
z^{1-\alpha}/(1-\alpha)$ is concave. For a concave function $g$, $g'(y) (y-x)
\leq g(y) - g(x)$, here we have $x=x_f(t)$ and $y=\rho_f$. Now, since $x_f(t)$ is optimal for the rate allocation problem of \FMSeq{enum:flow} and $\rho_f$ is feasible, we have that
\begin{equation}\label{eq:L-A-2}
n_f^\alpha(t) \frac{x_f^{1-\alpha}(t)}{1-\alpha} 
-
q_{\iota(f)}^\alpha(t) x_f(t)
\geq
n_f^\alpha(t) \frac{\rho_f^{1-\alpha}}{1-\alpha} 
-
q_{\iota(f)}^\alpha(t) \rho_f.
\end{equation}
Combining \eqref{eq:L-A-1} and \eqref{eq:L-A-2}, we have established
\eqref{eq:L-A}. Then, we can sum \eqref{eq:L-A} over all $f\in\Fscr$, to
obtain
\begin{equation}\label{eq:A-bound}
\Delta_n \leq \left[\Gamma \big(\rho - x(t)\big)\right]^\top q^\alpha(t).
\end{equation}

Now, consider the term $\Delta_q$ in \eqref{eq:dtL}. \FMSeq{enum:five},
\FMSeq{enum:six}, and \FMSeq{enum:idling} imply that
\begin{equation}\label{eq:B-1}
\dot{q}(t) = - (I - R^\top) \Pi \dot{s}(t) + \Gamma x(t).
\end{equation}
From \FMSeq{enum:seven} and \FMSeq{enum:packet},
\begin{equation}\label{eq:B-2}
\left[\big(I - R^\top\big) \Pi \dot{s}(t)\right]^\top q^\alpha(t)
= \sum_{\pi\in\Sscr}
\dot{s}_\pi(t) \pi^\top (I - R) q^\alpha(t)
=
\max_{\sigma\in\Sscr}\ \sigma^\top (I - R) q^\alpha(t).
\end{equation}
Together \eqref{eq:B-1} and \eqref{eq:B-2}, imply that
\begin{equation}\label{eq:B-bound}
\Delta_q \leq  
\left[ \Gamma x(t) \right]^\top q^\alpha(t)
  -
\max_{\sigma\in\Sscr}\ \sigma^\top (I - R) q^\alpha(t).
\end{equation}

Now, combining \eqref{eq:dtL}, \eqref{eq:A-bound}, and \eqref{eq:B-bound}, we
have that
\begin{equation}\label{eq:dtL-2}
\dt 
L_\alpha\big(n(t), q(t)\big)
\leq
(1+ \alpha)
\left(
\left[ \Gamma \rho \right]^\top q^\alpha(t)
  -
\max_{\sigma\in\Sscr}\ \sigma^\top (I - R) q^\alpha(t)
\right).
\end{equation} 
By the definition of $\Leff(\rho)\defeq \PRIMAL(\Xi\Gamma\rho)$, there exists some $s\in\R^\Sscr_+$ with $\1^\top s = \Leff(\rho)$ and 
\[
\Gamma \rho \leq \big(I - R^\top\big) \Pi s 
= \sum_{\pi \in \Sscr} s_\pi \big(I - R^\top\big) \pi.
\]
This implies that
\[
\left[\Gamma \rho\right]^\top q^\alpha(t) \leq 
\Leff(\rho) \max_{\sigma\in\Sscr}\ \sigma^\top (I - R) q^\alpha(t),
\]
Hence, by \eqref{eq:dtL-2},
\begin{equation}\label{eq:dtL-3}
\dt 
L_\alpha\big(n(t), q(t)\big)
\leq
  -
(1+\alpha) \big(1 - \Leff(\rho)\big)
\max_{\sigma\in\Sscr}\ \sigma^\top (I - R) q^\alpha(t).
\end{equation} 

In order to bound the right hand side of \eqref{eq:dtL-3}, we will argue that
\begin{equation}\label{eq:big-max}
\max_{\sigma\in\Sscr}\ \sigma^\top (I - R) q^\alpha(t)
\geq \frac{\1^\top q^\alpha(t)}{|\Escr|^2}.
\end{equation}
To see this, consider a queue $e_0\in\Escr$ with maximal length at time $t$,
i.e., $q_{e_0}(t) = \max_{e} q_e(t)$. Define $e_1,\dots, e_J\in\Escr$ to be a
sequence of distinct queues so that, for each $0 \leq i < J$, packets
departing from queue $e_i$ go to queue $e_{i+1}$, and packets departing from
queue $e_J$ exist the network. By our assumption of acyclic routing, such a
sequence exists, and since each queue is distinct, $J+1\leq |\Escr|$.
Consider the distinct schedules $\pi_0,\ldots,\pi_J\in\Sscr$, where, for
$0 \leq i \leq J$, the schedule $\pi_i$ serves exactly the queue $e_i$ ---
such schedules exist by the monotonicity assumption on the scheduling
constraints. Clearly, these schedules have weights given by
\[
\pi_i^\top (I - R) q^\alpha(t)
=
\begin{cases}
q^\alpha_{e_i}(t) - q^\alpha_{e_{i+1}}(t)
&
\text{if $0 \leq i < J$,} 
\\
q^\alpha_{e_J}(t) 
&
\text{if $i = J$.} 
\end{cases}
\]
Averaging over the $J+1$ schedules,
\[
\frac{1}{J+1} \sum_{i=0}^J \pi_i^\top (I - R) q^\alpha(t) 
= \frac{q^\alpha_{e_0}(t)}{J+1}.
\]
Since at least one schedule must have a weight that exceeds this average,
\[
\max_{\sigma\in\Sscr}\ \sigma^\top (I - R) q^\alpha(t)
\geq \frac{q^\alpha_{e_0}(t)}{J+1}
\geq \frac{\max_{e\in\Escr} q^\alpha_{e}(t)}{|\Escr|}
\]
Since 
\[
\max_{e\in\Escr} q^\alpha_{e}(t) \geq \frac{\1^\top q^\alpha(t)}{|\Escr|},
\]
\eqref{eq:big-max} follows.

Combining \eqref{eq:dtL-3} and \eqref{eq:big-max}, we obtain, when
$\Leff(\rho) \leq 1$,
\begin{equation}\label{eq:dtL-4}
\dt 
L_\alpha\big(n(t), q(t)\big)
\leq
  -
(1+\alpha) \big(1 - \Leff(\rho)\big)
\frac{\1^\top q^\alpha(t)}{|\Escr|^2} \leq 0.
\end{equation}
This establishes the first part of Lemma~\ref{lem:fluid-drain}.

\medskip \noindent To prove the second part of Lemma~\ref{lem:fluid-drain}, we
will consider two separate cases over initial conditions
$\big(n(0),q(0)\big)$ with $L_\alpha\big(n(0),q(0)\big)=1$: 
\begin{enumerate}[(i)]
\item\label{enum:q-big}
  $\1^\top q^{1+\alpha}(0) > \beps_1$.
\item\label{enum:q-small}
$\1^\top q^{1+\alpha}(0) \leq \beps_1$.
\end{enumerate}
Here, $\beps_1>0$ is a constant that will be determined shortly.

For case~(\ref{enum:q-big}), from the norm inequality in part~(\ref{enum:p-q}) of
Lemma~\ref{lem:norm}, we have that
\[
\beps_1^{\frac{1}{1+\alpha}} < \|q(0)\|_{1+\alpha} 
\leq \|q(0)\|_{\alpha},
\]
Thus, $\1^\top q^\alpha(0) > \beps_2 \defeq
\beps_1^{\frac{\alpha}{1+\alpha}}$. Due to \FMSeq{enum:one}, $q(\cdot)$ is
uniformly Lipschitz continuous, there exists $T_1>0$ such that $\1^\top
q^\alpha(t) \geq \beps_2/2$ for all $t\in [0,T_1]$. From \eqref{eq:dtL-4},
since $\Leff(\rho) < 1$,we have that
\[
\dt 
L_\alpha\big(n(t), q(t)\big)
\leq
  -
(1+\alpha) \big(1 - \Leff(\rho)\big)
\frac{\beps_2}{2|\Escr|^2} < 0,
\]
for all regular $t\in (0,T_1]$. Therefore, there exists $\delta_1 > 0$ so that
\begin{equation}\label{eq:delta-1}
L_\alpha\big(n(t), q(t)\big) \leq 1- \delta_1, 
\end{equation}
for all $t \geq T_1$.

For case~(\ref{enum:q-small}), the argument is more complicated. The basic
insight is as follows: if $\1^\top q^{1+\alpha}(0)$ is \emph{small}, then
$\Delta_q$ in \eqref{eq:dtL} must be small as well. Further, a good fraction
of the flows will be allocated the maximal rate (i.e., $C$), therefore
$\Delta_n$ in \eqref{eq:dtL} will be significantly negative.  This will leads
to strictly negative drift in $L$. We will now formalize this intuition and
make an appropriate choice of $\beps_1 > 0$ along the way.

First, by the Lipschitz continuity of $q(\cdot)$, there exists some $\tau_1>0$
such that, for all $t\in [0,\tau_1]$, $\1^\top q^{1+\alpha}(t) \leq 2\beps_1$.
Note that from \eqref{eq:B-bound}, the fact that $\bzero\in\Sscr$, and $x(t)
\leq C\1$, for all regular $t\in (0,\tau_1]$,
\[
\Delta_q \leq \left[ \Gamma x(t) \right]^\top q^\alpha(t) \leq C \1^\top q^{\alpha}(t).
\]
By Jensen's inequality,
\[
\left(\frac{1}{|\Fscr|} \1^\top q^{\alpha}(t)\right)^{\frac{1+\alpha}{\alpha}}
\leq
\frac{1}{|\Fscr|} \1^\top q^{1+\alpha}(t).
\]
Therefore, for regular $t\in (0,\tau_1]$,
\begin{equation}\label{eq:Dq}
\Delta_q \leq C |\Fscr|^{\frac{1}{1+\alpha}} \left(2\beps_1\right)^{\frac{\alpha}{1+\alpha}}.
\end{equation}

Again using the fact that $\1^\top q^{1+\alpha}(t) \leq 2\beps_1$ for $t\in
[0,\tau_1]$, it follows that $q_e(t) \leq \beps_3 \defeq (2
\beps_1)^{1/1+\alpha}$, for all $e \in \Escr$. Now, consider the set
\[
\Fscr' \defeq \{ f\in\Fscr\ :\ n_f(0) \geq 4 C \beps_3 \}.
\]
By Lipschitz continuity, there exists some $0 < T_2 < \tau_1$ such that,
for all $f\in\Fscr'$ and all $t\in [0,T_2]$, $n_f(t) \geq 2 C
\beps_3$. Then, by \FMSeq{enum:flow}, $x_f(t) = C$ for  regular $t\in
(0,T_2]$. Thus, if $f\in\Fscr'$, we have, for regular $t\in
(0,T_2]$,
\begin{equation}\label{eq:Dn-1}
\frac{n_f^{\alpha}(t)}{\mu_f \rho_f^\alpha}
\dot{n}_f(t)  
=
\frac{n_f^{\alpha}(t)}{\mu_f \rho_f^\alpha}
\big(\rho_f - x_f(t)\big)
\leq
-
\frac{n_f^{\alpha}(t)}{\mu_f \rho_f^\alpha}
(C - \rho_f)
\leq
- \beta_1 n_f^\alpha(t),
\end{equation}
where we define
\[
\beta_1 \defeq \min_{f\in\Fscr} \frac{C-\rho_f}{\mu_f \rho_f^\alpha} > 0.
\]
Finally, if $f \notin \Fscr'$, we have, for regular $t\in (0,T_2]$,
\begin{equation}\label{eq:Dn-2}
\frac{n_f^{\alpha}(t)}{\mu_f \rho_f^\alpha}
\dot{n}_f(t)  
=
\frac{n_f^{\alpha}(t)}{\mu_f \rho_f^\alpha}
\big(\rho_f - x_f(t)\big)
\leq
\frac{(2 C)^\alpha \rho_f^{1-\alpha}
  \beps_1^{\frac{\alpha}{1+\alpha}}}
{\mu_f}
\leq
\beta_2 (2 C)^\alpha 
  \beps_1^{\frac{\alpha}{1+\alpha}},
\end{equation}
where we define
\[
\beta_2 \defeq \max_{f\in\Fscr} \frac{\rho_f^{1-\alpha}}{\mu_f}.
\]

Using \eqref{eq:dtL} and \eqref{eq:Dq}--\eqref{eq:Dn-2}, it follows that, for
all $t\in[0,T_2]$,
\begin{equation}\label{eq:dtL-5}
\dt 
L_\alpha\big(n(t), q(t)\big)
\leq
(1+\alpha)
\left(
\left[
C |\Fscr|^{\frac{1}{1+\alpha}} 
+
\beta_2 |\Fscr| (2 C)^\alpha 
\right]
\beps_1^{\frac{\alpha}{1+\alpha}}
- \beta_1 \sum_{f\in\Fscr'} n_f^\alpha(t)
\right).
\end{equation}
From \eqref{eq:dtL-4}, for all $t \geq 0$, $L_\alpha\big(n(t), q(t)\big) \leq
1$. For $t\in [0,T_2]$, $\1^\top q(t)^{1+\alpha} \leq 2 \beps_1$. Suppose
that $\beps_1 < 1/4$. Then,
\[
\sum_{f \in \Fscr'} n_f^{1+\alpha}(t)  
\geq 
\frac{1 - 2 \beps_1}
{
  \max_{f\in \Fscr} \frac{1}{\mu_f \rho_f^\alpha}
}
\geq
\frac{1}
{
  4\max_{f\in \Fscr} \frac{1}{\mu_f \rho_f^\alpha}
}.
\]
for $t\in [0,T_2]$.
From part~(\ref{enum:p-q}) of Lemma~\ref{lem:norm},
\[
\sum_{f \in \Fscr'} n_f^{\alpha}(t) 
\geq \left( \frac{1} { 4\max_{f\in
      \Fscr} \frac{1}{\mu_f \rho_f^\alpha} } \right)^{\frac{\alpha}{1+\alpha}}
\defeq \beta_3 > 0.
\]
Then, for regular $t\in (0,T_2]$,
\[
\dt 
L_\alpha\big(n(t), q(t)\big)
\leq
-(1+\alpha) \beta_1\beta_3/2 < 0.
\]
It follows that there exists $\delta_2 > 0$ such that, for all $t \geq T_2$,
\begin{equation}\label{eq:delta-2}
  L_\alpha\big(n(t), q(t)\big)  \leq  1-\delta_2.
\end{equation}

Lemma~\ref{lem:fluid-drain} follows from \eqref{eq:delta-1} and
\eqref{eq:delta-2} with $\delta_* = \min\{\delta_1, \delta_2\}$ and $T_* =
\max\{T_1, T_2\}$.

\section{Critical Loading}
\label{sec:critical-loading}

We have established the throughput optimality of the system under the \MWUMa
control policy, for any $\alpha\in (0,\infty)\setminus\{1\}$. Thus, this
entire family of policies possesses good \emph{first order}
characteristics. Further, there may be many other throughput optimal policies
outside the class of \MWUMa policies.  This naturally raises the question of
whether there is a `best' choice of $\alpha$, and how the resulting \MWUMa
policy might compare against the universe of all other policies.

In order to answer these questions, we desire a more refined analysis of
policy performance than throughput optimality. One way to obtain such an
analysis is via the study of a \emph{critically loaded system}, i.e., a system
with critically admissible arrival rates. Under a critical loading, fluid
model solutions take non-trivial values over entire horizon. In contrast, for
strictly admissible systems under throughput optimal policies, all fluid
trajectories go to $\bzero$ (cf.\ Lemma~\ref{lem:fluid-drain}). We will employ
the study of the fluid model solutions of critically loaded systems as a tool
for the comparative analysis of network control policies.

In particular, given a vector of flow counts, $n\in \Rp^{\Fscr}$, and the
vector of queue lengths, $q\in \Rp^{\Escr}$, define the linear cost function
\begin{equation}\label{eq:cost}
c(n,q) \defeq 
\sum_{f \in \Fscr} \frac{n_f}{\mu_f}
+ \sum_{e\in \Escr} q_e = \1^\top 
\left[\Gamma \diag(\mu^{-1}) n + q\right].
\end{equation}
This cost function is analogous to a `minimum delay' objective in a
packet-level queueing network: a cost is incurred for each queued packet, and
a cost is also incurred for each outstanding flow, proportional to the number
of packets that it is expected to generate.


In this section, we establish fundamental lower bounds that apply to the cost
incurred in a critically loaded fluid model under \emph{any} scheduling
policy. In Sections~\ref{sec:balanced} and \ref{sec:invariant}, we will
compare these with the costs incurred by \MWUMa control policies. We shall
find that as $\alpha \to 0^+$, the cost induced by the \MWUMa
algorithms improves and becomes close to the algorithm independent lower bound

\subsection{Virtual Resources and Workload}

We begin with some definitions. First, consider the dual of the LP
$\PRIMAL(\lambda)$,
\[
\begin{array}{l@{\ }lll}
\DUAL(\lambda) \defeq
&
\maximize_\zeta & \lambda^\top \zeta \\
& \subjectto
&
\Pi^\top \zeta  \leq \1, 
\\[2pt]
& &
\zeta\in\Rscr^\Escr_+.
\end{array}
\]
Note that there is no duality gap, thus the value of $\PRIMAL(\lambda)$ is
equal to the value of $\DUAL(\lambda)$.

\begin{definition}[\textbf{Virtual Resource}] 
  We will call any feasible solution $\zeta \in \Rp^{\Escr}$ of dual
  optimization problem $\DUAL(\Xi\Gamma \rho)$ a \textbf{virtual resource}.
  Suppose the system is critically loaded, i.e., the offered load vector
  $\rho$ satisfies
  \[
  \Leff(\rho) = \PRIMAL(\Xi \Gamma \rho) = \DUAL(\Xi \Gamma \rho) = 1.
  \]
  Then, we call a virtual resource that is an optimal solution of
  $\DUAL(\Xi \Gamma \rho)$ a \textbf{critical virtual resource}.
\end{definition}

For a critically loaded system with offered load vector $\rho$, let
$\CR(\rho)$ denote the set of all critical virtual resources. Note that
$\CR(\rho)$ is a bounded polytope and hence possesses finitely many extreme
points. 
Let $\CRE(\rho)$ denote the set of extreme points of $\CR(\rho)$.

The following definition captures the amount of `work' associated with a
critical resource, as a function of the current state of the system.
\begin{definition}[\textbf{Workload}]
  Consider a critically loaded system with an offered load vector $\rho$ and a
  critical virtual resource $\zeta\in\CR(\rho)$. If the flow count and
  queue length vectors are given by $(n,q)$, the \textbf{workload} associated
  with the resource $\zeta$ is defined to be
  \[
  w_\zeta(n,q) \defeq \zeta^\top \Xi \left[ q + \Gamma \diag(\mu)^{-1} n\right].
  \]
\end{definition}

\subsection{A Lower Bound on Fluid Trajectories}

Consider a critically loaded system with offered load vector $\rho$. We claim
that the following fundamental lower bound holds on the fluid trajectory under
\emph{any} algorithm. This bound can be thought of as a minimal
work-conservation requirement.

\begin{lemma}\label{lem:wc}
  Consider the fluid model trajectory of system under any scheduling and rate
  allocation policy, with flow count and queue length processes given by
  $\big(n(\cdot),q(\cdot)\big)$. Then, for any time $t \geq 0$ and any critical
  virtual resource $\zeta \in \CR(\rho)$,
  \begin{equation}\label{eq:wc}
    w_\zeta\big(n(0),q(0)\big)
    \leq
    w_\zeta\big(n(t),q(t)\big).
  \end{equation}
\end{lemma}
\begin{proof}
  Given a time interval $[0,T]$, for any $T > 0$, consider the fluid model
  trajectory $\zscr(\cdot)$ of the form \eqref{eq:fms-path}.  By
  Theorem~\ref{thm:fms}, this fluid trajectory must satisfy the algorithm
  independent fluid model equations \FMSeq{enum:one}--\FMSeq{enum:eight}.  By
  \FMSeq{enum:one}, the trajectory is Lipschitz continuous and differentiable
  for almost all $t \in (0,T)$.  For any such regular point $t$, by
  \FMSeq{enum:two}--\FMSeq{enum:four}, we have
  \[
  \dot{n}(t) = \nu - \diag(\mu) \dot{\bar{x}}(t). 
  \]
  Thus,
  \begin{equation}\label{eq:lb1}
    \Gamma \diag(\mu^{-1}) \dot{n}(t)  =  \Gamma \rho -  \Gamma \dot{\bar{x}}(t).
  \end{equation}
  From \FMSeq{enum:five}--\FMSeq{enum:six}, we obtain
  \begin{equation}\label{eq:lb2}
    \dot{q}(t) = 
    \left(I-R^\top\right) \dot{z}(t) 
    - \left(I-R^\top\right) \Pi \dot{s}(t)
    +
    \Gamma \dot{\bar{x}}(t).
  \end{equation}
  Adding \eqref{eq:lb1} and \eqref{eq:lb2}, we obtain
  \[
    \dot{q}(t) +  \Gamma \diag(\mu^{-1}) \dot{n}(t)  
    =  \Gamma \rho + \left(I-R^\top\right) \left(\dot{z}(t) 
    - \Pi \dot{s}(t)\right). 
  \]
  Now, multiplying both sides by $\Xi\defeq \left(I-R^\top\right)^{-1}$, we
  obtain
  \begin{equation}\label{eq:lb4}
    \Xi \left[\dot{q}(t) +  \Gamma \diag(\mu^{-1}) \dot{n}(t)\right] =
    \Xi \Gamma \rho + \dot{z}(t) - \Pi \dot{s}(t). 
  \end{equation}
  Consider a critical virtual resource $\zeta \in \CR(\rho)$.  Since $\zeta$
  is $\DUAL(\Xi \Gamma \rho)$ optimal, $\zeta^\top \Xi \Gamma \rho = 1$.
  Taking an inner product of \eqref{eq:lb4} with $\zeta$, we obtain
  \begin{equation}\label{eq:lb6}
    \zeta^\top \Xi \left[\dot{q}(t) +  \Gamma \diag(\mu^{-1}) \dot{n}(t)\right] 
    =
    1 + \zeta^\top \dot{z}(t) - \zeta^\top \Pi \dot{s}(t). 
  \end{equation}
  Now, by \FMSeq{enum:eight}, $z(\cdot)$ is non-decreasing. Since $\zeta$ is
  non-negative, then $\zeta^\top \dot{z}(t) \geq 0$.  By \FMSeq{enum:eight},
  $\dot{s}(t)$ is non-negative. Since $\zeta$ is $\DUAL(\Xi \Gamma \rho)$
  feasible and from \FMSeq{enum:seven}, it follows that $ \zeta^\top \Pi
  \dot{s}(t) \leq \bone^\top \dot{s}(t) = 1.  $ Applying these observations to
  \eqref{eq:lb6}, it follows that
  \[
  \dt w_\zeta\big(n(t),q(t)\big)
  = \zeta^\top \Xi \left[\dot{q}(t) +  \Gamma \diag(\mu^{-1}) \dot{n}(t)\right] 
  \geq 0.
  \]
  Given that $\big(n(\cdot),q(\cdot)\big)$ is Lipschitz continuous, the
  desired result follows immediately.
\end{proof}

Lemma~\ref{lem:wc} guarantees the conservation of workload under any
policy. This motivates the \emph{effective cost} of a state
$(n,q)\in\Rp^\Fscr\times\Rp^\Escr$, defined by the linear program
\begin{equation}\label{eq:eff-cost}
\begin{array}{l@{\ }lll}
c^*(n,q) \defeq
&
\minimize_{n',q'} & c(n',q') \\
& \subjectto 
& w_\zeta(n',q') \geq w_\zeta(n,q),
~\forall\ \zeta\in\CRE(\rho),
\\[2pt]
& &
n\in\Rp^\Fscr,\quad q\in\Rp^\Escr.
\end{array}
\end{equation}
The effective cost is the lowest cost of any state with at least as much
workload as $(n,q)$. We have the following lower bound on the cost achieved
under any fluid trajectory:

\begin{theorem}\label{thm:eff-cost}
  Consider fluid model trajectory of system under any scheduling and rate
  allocation policy, with flow count and queue length processes given by
  $\big(n(\cdot),q(\cdot)\big)$. Then, for any time $t \geq 0$, the
  instantaneous cost $c\big(n(t),q(t)\big)$ is bounded below according to
  \begin{equation}\label{eq:cost-lower-bound}
    c^*\big(n(0),q(0)\big) \leq c\big(n(t),q(t)\big).
  \end{equation}
\end{theorem}
\begin{proof}
  By Lemma~\ref{lem:wc}, if the initial condition of a fluid trajectory
  satisfies $\big(n(0),q(0)\big)=(n,q)$, then $\big(n(t),q(t)\big)$ is
  feasible for \eqref{eq:eff-cost} for every $t \geq 0$. The result
  immediately follows.
\end{proof}

\section{Balanced Systems}
\label{sec:balanced}

In this section, we will develop a bound on the cost achieved in a fluid model
solution under the \MWUMa policy. In particular, we will establish that this
cost, at any instant of time, is within a constant factor of the cost
achievable under \emph{any} policy. The constant factor is uniform across the
entire fluid trajectory, and relates to a notion of `balance' of the critical
resources of the network that we will describe shortly.

We begin with a preliminary lemma, that provides an upper bound on the cost
under the \MWUMa policy. This upper bound is closely related to the Lyapunov
function introduced earlier for studying the system stability.

\begin{lemma}\label{lem:ub}
  Consider a fluid model trajectory of system under the \MWUMa policy, and
  denote the flow count and queue length processes by
  $\big(n(\cdot),q(\cdot)\big)$. Suppose that the offered load vector $\rho$
  satisfies $\Leff(\rho) \leq 1$. Then, at any time $t \geq 0$, it must be
  that
  \begin{equation}\label{eq:lub}
    c\big(n(t),q(t)\big) 
    \leq 
    \big(1+ \beta(\alpha)\big)
    c\big(n(0),q(0)\big),
  \end{equation}
  where $\beta(\alpha) \to 0$ as $\alpha \to 0^+$. 
\end{lemma}
\begin{proof}
  Recall the Lyapunov function $L_\alpha$ from \eqref{eq:lyapunov}. It follows
  from Lemma~\ref{lem:fluid-drain} that, so long as $\Leff(\rho) \leq 1$,
  \begin{equation}\label{eq:lub0}
    L_\alpha\big(n(t), q(t)\big) \leq L_\alpha\big(n(0), q(0)\big).
  \end{equation}
  Applying Lemma~\ref{lem:norm}, with $p\defeq
  1+\alpha$ and $d \defeq |\Escr| + |\Fscr|$,
  \[
    \sum_{f \in \Fscr} 
    \frac{n_f(t)}{\mu_f} 
    \left(\frac{1}{\nu_f}\right)^{\frac{\alpha}{1+\alpha}}
    + \sum_{e\in \Escr} q_e(t)
    \leq
    d^{\frac{\alpha}{1+\alpha}} 
    \left[\sum_{f \in \Fscr} 
      \frac{n_f(0)}{\mu_f} 
      \left(\frac{1}{\nu_f}\right)^{\frac{\alpha}{1+\alpha}} 
      + \sum_{e\in \Escr} q_e(0) \right].
  \]
  Now, as $\alpha\to 0^+$, $d^{\frac{\alpha}{1+\alpha}} \to 1$. 
  Also,
  \begin{equation}\label{eq:lub5}
    \left(\frac{1}{\nu^*}\right)^{\frac{\alpha}{1+\alpha}}
    \leq
    \left(\frac{1}{\nu_f}\right)^{\frac{\alpha}{1+\alpha}}
    \leq
    \left(\frac{1}{\nu_*}\right)^{\frac{\alpha}{1+\alpha}},
  \end{equation}
  where $\nu_* \defeq \min_f \nu_f$ and $\nu^* \defeq \max_f \nu_f$. 
  Thus, as $\alpha\to 0^+$, $1/\nu_f\tends 1$ uniformly over $f$.
  The result then follows.
\end{proof}

The following definition is central to our performance guarantee:
\begin{definition}[\textbf{Balance Factor}]
  Given a system that is critically loaded with offered load vector $\rho$,
  define the
  \textbf{balance factor} as the value of the optimization
  problem
  \[
  \begin{array}{l@{\ }lll}
    \gamma(\rho) \defeq
    &
    \minimize_{n,q,n',q'}
    & 
    \displaystyle c(n',q')
    \\
    &
    \subjectto
    & 
    w_\zeta(n',q')
    \geq
    w_\zeta(n,q),\ 
    \forall\ \zeta\in\CRE(\rho), \\[2pt]
    & &
    c(n,q) = 1, \\[2pt]
    & & 
    n,n'\in\Rp^\Fscr,
    \quad
    q,q'\in\Rp^\Escr.
  \end{array}
  \]
\end{definition}

It is clear that $\gamma(\rho)\geq 0$, since $n',q'\geq \bzero$. Since there
are feasible solutions with $(n,q)=(n',q')$, it is also true that
$\gamma(\rho)\leq 1$. In order to interpret $\gamma(\rho)$, assume for the
moment that there is only a single critical extreme resource
$\zeta\in\CRE(\rho)$. If we define $v \defeq \Xi^\top \zeta$, then the
constraint that $w_\zeta(n',q') \geq w_\zeta(n,q)$ is equivalent to
\[
v^\top\left[ \Gamma \diag(\mu)^{-1} n' + q' \right]
\geq
v^\top\left[ \Gamma \diag(\mu)^{-1} n + q \right].
\]
In this case, it is clear that the solution to the LP defining $\gamma(\rho)$ is
given by
\[
\gamma(\rho) = 
(\min_e v_e)/(\max_e v_e).
\]
Hence, $\gamma(\rho)$ is the measure of the degree of `balance' of the
influence of the critical resource $\zeta$ across buffers in the network.

In the more general case (i.e., $|\CRE(\rho)|\geq 1$), define the set $\Vscr
\defeq \text{span}\ \{\Xi^\top \zeta\ :\ \zeta\in\CRE(\rho)\}$. It is not
difficult to see that $\gamma(\rho)>0$ if and only if, for each queue $e \in
\Vscr$, there exists some $v\in\Vscr$ with $v_e > 0$, i.e., if every queue is
influenced by \emph{some} critical resource. We call networks where
$\gamma(\rho) > 0$ \emph{balanced}. In the extreme, if $\1 \in \Vscr$, then
$\gamma(\rho) = 1$.

The following is the main theorem of this section. It offers a bound on the
cost incurred at any instant in time under the \MWUMa policy, relative that
incurred under \emph{any} other policy. This bound is a function of the
balance factor.
\begin{theorem}\label{thm:balance}
  Consider fluid model trajectory of a critically loaded system under the
  \MWUMa policy and denote the flow count and queue length processes by
  $\big(n(\cdot),q(\cdot)\big)$. Suppose that $\gamma(\rho) > 0$. Let
  $\big(n'(\cdot),q'(\cdot)\big)$ be the flow count and queue length policies
  under an arbitrary policy given the same initial conditions, i.e.,
  $n(0)=n'(0)$ and $q(0)=q'(0)$. Then, at any time $t\geq 0$, it must be that
  \begin{equation}\label{eq:balance-bound}
    c\big(n(t),q(t)\big)
    \leq
    \frac{1+\beta(\alpha)}{\gamma(\rho)}
    c\big(n'(t),q'(t)\big),
  \end{equation}
  where $\beta(\alpha) \to 0$ as $\alpha \to 0^+$. 
\end{theorem}
\begin{proof}
  First, note that if $\big(n(0),q(0)\big)=\bzero$, i.e., the system is empty,
  then this holds for all $t \geq 0$ (cf.\ Theorem~\ref{thm:invariant}). In
  this case, \eqref{eq:balance-bound} is immediate. Otherwise, fix $t \geq 0$,
  and set $\bar{c} \defeq c\big(n(0),q(0)\big) > 0$. Define
  \[
  (n',q') \defeq \big(n'(t),q'(t))/\bar{c},\quad
  (n,q) \defeq \big(n(0),q(0))/\bar{c}.
  \]
  Using Lemma~\ref{lem:wc}, it is clear that $(n,q,n',q')$ is feasible for
  the LP defining $\gamma(\rho)$.  Thus,
  \[
  c\big(n(0),q(0)\big)
  \leq
  \frac{1}{\gamma(\rho)} c\big(n'(t),q'(t)\big).
  \]
  The result then follows by applying Lemma~\ref{lem:ub}.
\end{proof}

\section{Invariant Manifold}\label{sec:invariant}

In Section~\ref{sec:balanced}, we proved a constant factor guarantee on the
cost of the \MWUMa policy, relative to the cost achieved under any other
policy. Our bound held point-wise, at \emph{every} instant of time. However,
the constant factor of the bound depends on the balance factor, and this could
be very large.

In this section, we consider a different type of analysis. Instead of
considering the evolution of the fluid model for \emph{every} time $t$, we
instead examine the asymptotically limiting states of the fluid model as
$t\to\infty$. In particular, we characterize these \emph{invariant} states as
fixed points in the solution space of an optimization problem. We shall also
show that these fixed points are attractive, i.e., starting from any initial
condition, the fluid trajectory reaches an invariant state. We will quantify
time to converge to the invariant manifold as a function of the initial
conditions of the fluid trajectory.

This characterization of invariant states is key towards establishing the
\emph{state space collapse} property of the system under a heavy traffic limit
\cite{Bramson98}. Moreover, we shall demonstrate that these invariant states
are cost optimal as $\alpha\to 0^+$. In other words, the cost of an invariant
state cannot be improved by \emph{any} policy.

\subsection{Optimization Problems}\label{sec:opt-prob}

We start with two useful optimization problems that will be useful in
characterizing invariant states of the fluid trajectory. We assume that the
system is critically loaded.

Suppose we are given a state $(n,q)\in\Rp^\Fscr\times\Rp^\Escr$ of,
respectively, flow counts and queue lengths. Define the optimization problem
\[
\begin{array}{l@{\ }lll}
\ALGP(n,q) \defeq 
& \minimize_{n',q',t,x,\sigma} & L_\alpha(n',q') \\
& \subjectto & n' = n + t\big[\nu - \diag(\mu) x\big], \\[2pt]
& & q' = q + t\left[\Gamma x - \left(I - R^\top\right) \sigma\right], \\[2pt]
& & 
n' \in\Rp^\Fscr,\quad
q'\in\Rp^\Escr,\quad
t\in\Rp,
\\[2pt]
& &
x\in [0,C]^\Fscr,\ 
\sigma\in\Lambda.
\end{array}
\]
Here, recall that $\Lambda$ is the scheduling capacity region of the network,
defined by \eqref{eq:Lambda}. Similarly, define the optimization problem
\[
\begin{array}{l@{\ }lll}
  \ALGD(n,q) \defeq
  & \minimize_{n',q'} & L_\alpha(n',q')
  \\
  & \subjectto 
  & w_\zeta(n',q') \geq w_\zeta(n,q),
  \\[2pt]
  & & \quad \forall\ \zeta\in\CRE(\rho),
  \\[2pt]
  & & n'\in\Rp^\Fscr,\quad q'\in\Rp^\Escr.
\end{array}
\]

Intuitively, given a state $(n,q)$, $\ALGP(n,q)$ finds a state $(n',q')$ which
minimizes the Lyapunov function $L_\alpha$ and can be reached starting from
$(n,q)$, using feasible scheduling and rate allocation
decisions. $\ALGP(n,q)$, on the other hand, finds a state $(n',q')$ which
minimizes the Lyapunov function and has at least as much workload as $(n,q)$.
The following result states that $\ALGP(n,q)$ and $\ALGD(n,q)$ are equivalent
optimization problems:

\begin{lemma}\label{lem:equivopt}
  A state $(n',q')\in\Rp^\Fscr\times\Rp^\Escr$ is feasible for the
  optimization problem $\ALGP(n,q)$ if and only if it is feasible for the
  optimization problem $\ALGD(n,q)$.
\end{lemma}
\begin{proof} 
  First, consider any $(n',q',t,x,\sigma)$ that is feasible for
  $\ALGP(n,q)$. Note that feasibility for
  $\ALGP(n,q)$ implies that
  \[
    \Gamma \diag(\mu)^{-1} n' + q' \geq
    \Gamma \diag(\mu)^{-1} n +   q + 
    t \Big[ \Gamma \diag(\mu)^{-1} \nu - \Gamma x \Big] + 
    t \left[ \Gamma x - \left(I-R^\top\right) \sigma\right].
  \]
  Therefore, if $\zeta \in \CRE(\rho)$, we have that
  \[
  w_\zeta(n',q') = w_\zeta(n,q) 
  + t \left[ \zeta^\top \Xi \Gamma\rho - \zeta^\top\sigma\right].
  \]
  Since $\sigma \in \Lambda$ and $\zeta$ is feasible for
  $\DUAL(\Xi\Gamma\rho)$, we have $\zeta^\top \sigma \leq 1$. Since $\zeta \in
  \CRE(\rho)$, we have $\zeta^\top \Xi \Gamma\rho = 1$.  Therefore, as $t \geq
  0$, it follows that
  \[
  w_\zeta(n',q') \geq w_\zeta(n,q).
  \]
  That is, $(n',q')$ is $\ALGD(n,q)$ feasible. 

  Next, assume that $(n',q')$ is feasible for $\ALGD(n,q)$. Given some $t \geq
  0$, define
  \[
  x \defeq \diag(\mu)^{-1} \left[\nu - t^{-1} (n'-n)\right],
  \ \ 
  \sigma \defeq \Xi \left[\Gamma x - t^{-1} (q'-q)\right].
  \]
  With these definitions, if we establish existence of $t\geq 0$ so that
  $\bzero \leq x \leq C\1$ and $\sigma\in\Lambda$, then $(n',q',t,x,\sigma)$
  is feasible for $\ALGP(n,q)$ feasible.

  Note that as $t\tends\infty$, $x\tends\rho$. By assumption, $0 < \rho_f <
  C$, for all $f\in\Fscr$. Therefore, for $t$ sufficiently large, $\bzero \leq
  x \leq C\1$.

  Next, we wish to show that, for $t$ sufficiently large, $\sigma \in
  \Lambda$. This requirement is equivalent to demonstrating that
  $\PRIMAL(\sigma) \leq 1$ and that $\sigma \geq \bzero$.  To show that
  $\PRIMAL(\sigma) \leq 1$, note that $\PRIMAL(\sigma) = \DUAL(\sigma)$ and
  suppose that $\zeta$ is feasible for $\DUAL(\sigma)$. Then,
  \[
  \begin{split}
    \zeta^\top \sigma & = 
    \zeta^\top \left[\Xi\Gamma x - t^{-1} (q'-q)\right] \\
    & = \zeta^\top\left[\Xi\Gamma\rho
      - t^{-1} \Xi\Gamma\diag(\mu)^{-1} (n'-n) 
      -  t^{-1} (q'-q)\right] \\
    & = \zeta^\top\Xi\Gamma\rho 
    - t^{-1} \left[w_\zeta(n',q') - w_\zeta(n,q)\right]. 
  \end{split}
  \]
  If $\zeta \in \CR(\rho)$, then
  \[
  \zeta^\top\Xi\Gamma\rho = 1,
  \quad\text{and}\quad
  w_\zeta(n',q') - w_\zeta(n,q) \geq 0,
  \]
  thus $\zeta^\top \sigma \leq 1$. On the other hand, if
  $\zeta\notin\CR(\rho)$, $\zeta^\top\Xi\Gamma\rho < 1$. Therefore, in any
  event, for $t$ sufficiently large, $\DUAL(\rho) \leq 1$.

  To show that $\sigma \geq \bzero$, note that
  \[
  \begin{split}
    \sigma & =  \Xi\left[\Gamma x - t^{-1} (q'-q)\right] \\
    & =  \Xi \Gamma \rho 
    - t^{-1} \left[ \Xi\Gamma \diag(\mu)^{-1} (n'-n) + \Xi (q'-q) \right].
  \end{split}
  \]
  By assumption,
  $\Xi\Gamma \rho > \bzero$. Therefore, for $t$ sufficiently large enough,
  $\sigma \geq \bzero$.
\end{proof}

\subsection{Fixed Points: Characterization} 

Note that the optimization problem $\ALGD(n,q)$ has a convex feasible set with
a strictly convex and coercive objective function (see, e.g.,
\cite{BertsekasNLP}). By standard arguments from theory of convex
optimization, it follows that an optimal solution exists and is unique. Hence,
we can make the following definition:
\begin{definition}[\textbf{Lifting Map}]
  Given a critically scaled system, we define the \textbf{lifting map}
  $\Delta\colon \Rp^\Fscr\times\Rp^\Escr \to \Rp^\Fscr\times\Rp^\Escr$ to be
  the function that maps a state $(n,q)$ to the unique solution of the
  optimization problem $\ALGD(n,q)$.
\end{definition}

The main result of this section is to characterize the invariant states of
fluid model as the fixed points of lifting map $\Delta$.
\begin{theorem}\label{thm:invariant}
  A state $(n,q) \in \Rp^{\Fscr}\times\Rp^\Escr$ is an invariant state of a
  fluid model solution under the \MWUMa policy if and only if it is a fixed
  point of $\Delta$, i.e.,
\[
(n,q) = \Delta(n,q).
\]
\end{theorem}
\begin{proof}
The proof follows by establishing equivalence of the
following statements, for every state $(n,q)$:
\begin{enumerate}[(i)]
  \itemsep=0pt
\item\label{enum:lm-i} $(n,q) = \Delta(n,q)$.
\item\label{enum:lm-ii} Any fluid model solution satisfying the initial
  condition $\big(n(0),q(0)\big) = (n,q)$ has $\big(n(t),q(t)\big) = (n,q)$
  for all $t$.
\item\label{enum:lm-iii} There exists a fluid model solution with
  $\big(n(t),q(t)\big) = (n,q)$ for all $t$.
\item\label{enum:lm-iv} $(n,q)$ satisfy
  \begin{gather}
    (\Gamma \rho)^\top q^\alpha = \max_{\pi \in \Sscr} \pi^\top(I-R)q^\alpha,
    \label{eq:lm-iv-i}
    \\
    \rho_f q_{\iota(f)} = n_f,\quad \forall\ f\in\Fscr.
    \label{eq:lm-iv-ii}
  \end{gather}
\end{enumerate}


\medskip


\noindent
\emph{(\ref{enum:lm-i}) $\Rightarrow$ (\ref{enum:lm-ii}):} If $(n,q) =
\Delta(n,q)$, then it solves $\ALGD(n,q)$. Consider a fluid model solution
with an initial state $\big(n(0), q(0)\big) = (n,q)$. By
Lemma~\ref{lem:fluid-drain}, it follows that, for all $t$, $L_\alpha\big(n(t),
q(t)\big) \leq L_\alpha(n,q)$.  From the fluid model equations
\FMSeq{enum:one}--\FMSeq{enum:idling}, $\big(n(t), q(t)\big)$ is $\ALGP(n,q)$
feasible, for all $t$.  Therefore, it follows that $\big(n(t), q(t)\big)$ is
an optimal solution of $\ALGP(n,q)$, and, by Lemma~\ref{lem:equivopt}, of
$\ALGD(n,q)$.  Since $\ALGD(n,q)$ has $(n,q)$ as its unique solution, it
follows that $\big(n(t),q(t)\big) = (n,q)$, for all $t$.

\medskip

\noindent
\emph{(\ref{enum:lm-ii}) $\Rightarrow$ (\ref{enum:lm-iii}):} This follows in a
straightforward manner by considering the arguments in Theorem~\ref{thm:fms}
with initial conditions given by $(n,q)$.

\medskip

\noindent
\emph{(\ref{enum:lm-iii}) $\Rightarrow$ (\ref{enum:lm-iv}):} Consider a fluid
model solution that satisfies $\big(n(t),q(t)\big)=(n,q)$, for all $t$. Then,
for any regular point $t$, we have $\dot{n}(t) = \bzero$ and
$\dot{q}(t) = \bzero$.  Using \FMSeq{enum:two}--\FMSeq{enum:four}, it follows
that $x(t) \defeq \dot{\bar{x}}(t) = \rho$.  For any $f\in\Fscr$, if
$n_f = n_f(t) > 0$ and $x_f(t) = \rho_f < C$, then by \FMSeq{enum:flow} it
must be that $x_f(t) = n_f(t)/q_{\iota(f)}(t)$.  Therefore, $\rho_f q_{\iota(f)} = n_f
$. Similarly, if $n_f = 0$, it must be that $q_{\iota(f)}=0$ by
\FMSeq{enum:nzero}.

Now, define $H(t) \defeq \bone^\top q^{1+\alpha}(t)$.  Since $q(\cdot)$ is
constant, applying \FMSeq{enum:five}, \FMSeq{enum:six}, \FMSeq{enum:idling},
it must be that for every regular $t$,
\[
0 
=  \dot{H}(t) 
= \dot{q}(t)^\top q^\alpha(t)
= \left[
  \Gamma \rho
  -
  \left(I - R^\top\right) \Pi \dot{s}(t) 
\right]^\top q^\alpha(t).
\]
Applying \FMSeq{enum:seven} and \FMSeq{enum:packet},
\[
0 
= 
  (\Gamma \rho)^\top q^\alpha
  -
  \max_{\pi\in\Sscr}\ \pi^\top (I - R) q^\alpha.
\]

\medskip

\noindent
\emph{(\ref{enum:lm-iv}) $\Rightarrow$ (\ref{enum:lm-i}):} Suppose $(n,q)$
satisfy \eqref{eq:lm-iv-i}--\eqref{eq:lm-iv-ii}. Define $(n',q') \defeq
\Delta(n,q)$.  Since $(n',q')$ solves the optimization problem $\ALGD(n,q)$,
by Lemma~\ref{lem:equivopt}, there exists $(t,x,\sigma)$ so that
$(n',q',t,x,\sigma)$ is an optimal solution for $\ALGP(n,q)$. This solution
must satisfy
\[
n' = n + t\big[\nu - \diag(\mu) x\big],\quad
q' = q + t\left[\Gamma x - \left(I - R^\top\right) \sigma\right].
\]
Now consider the trajectory
\[
\big(n(\tau), q(\tau)\big) \defeq (n,q) +  \frac{\tau}{t} (n'-n,q'-q),
\quad\forall\ \tau\in [0,t].
\]
Define $J$ to be the Lyapunov function $L_\alpha$ evaluated along this path,
i.e., $J(\tau) \defeq L_\alpha\big(n(\tau),q(\tau)\big)$.  Then, 
\[
\begin{split}
\frac{\dot{J}(0)}{1+\alpha}
& = 
\sum_{f\in\Fscr}
\frac{n_f^\alpha(\nu_f - \mu_f x_f)}{\mu_f \rho_f^\alpha} 
+ (\Gamma x)^\top q^\alpha
- \sigma^\top (I-R)q^\alpha  
\\
& = 
\underbrace{
\Bigg(
\sum_{f\in\Fscr}
\frac{n_f^\alpha(\nu_f - \mu_f x_f)}{\mu_f \rho_f^\alpha} 
+ (\Gamma \delta)^\top q^\alpha
\Bigg)
}_{\mbox{\normalsize $X$}} 
+ 
\underbrace{
\Bigg(
(\Gamma \rho)^\top q^\alpha
- \sigma^\top (I-R)q^\alpha  
\Bigg)
}_{\mbox{\normalsize $Y$}},
\end{split}
\]
where $\delta \defeq x - \rho$.

First, consider $Y$. Since $\sigma \in \Lambda$, there exists some
$s\in\Rp^\Sscr$ with $\1^\top s \leq 1$ and $\sigma \leq \Pi s$. From the
monotonicity of $\Sscr$, we can pick $s$ so that $\sigma = \Pi s$.  Therefore,
\[
\sigma^\top (I-R) q^\alpha 
=
s^\top \Pi^\top (I - R) q^\alpha
\leq 
\max_{\pi \in \Sscr}\ \pi^\top(I-R)q^\alpha.
\]
Then, by \eqref{eq:lm-iv-i}, it follows that $Y \geq 0$. Now, consider $X$,
and note that $X = 0$ by \eqref{eq:lm-iv-ii} along with
the fact that
\begin{equation}\label{eq:fp2}
X 
= 
\sum_{f\in\Fscr} \left(
\frac{n_f^\alpha (\rho_f - x_f)}{\rho_f^\alpha}  
+ \delta_f q^\alpha_{\iota(f)} 
\right)
=
\sum_{f\in\Fscr} \delta_f
\left(
q^\alpha_{\iota(f)} 
- \frac{n_f^\alpha}{\rho_f^\alpha}  
\right).
\end{equation}

Thus, we have that $\dot{J}(0) \geq 0$. Since $J(\tau)$ is a convex function,
this implies that $J(0) \leq J(t)$, i.e., $L_\alpha(n,q) \leq
L_\alpha(n',q')$.  Due to uniqueness of the optimal solution to $\ALGD(n,q)$,
it follows that $(n',q') = (n,q)$.
\end{proof}

\subsection{Fixed Points: Attractiveness}

We will now establish the attractiveness of the space of fixed points.
Specifically, we will show that starting from any initial state, the fluid
trajectory converges (arbitrarily close to) space of fixed points, in
finite time.

Given $\beps > 0$, define 
\[
\Jscr_\beps \defeq 
\left\{ 
(n, q)\in\Rp^\Fscr\times\Rp^\Escr\ 
:\ 
\| (n,q) - \Delta(n,q) \|_1 < \beps
\right\}.
\]
In other words, $\Jscr_\beps$ is the set of states $(n,q)$ which are
$\beps$-approximate fixed points (in an $\ell_1$-norm sense) of the lifting
map. Given a fluid trajectory $\big(n(\cdot), q(\cdot)\big)$,
define
\[
h_\beps\big(n(\cdot),q(\cdot)\big) \defeq
\inf\ \big\{ t \geq 0\ :\ 
\big(n(s),q(s)\big) \in \Jscr_\beps,\ \forall\ s \geq t \big\}.
\]
In other words, $h_\beps\big(n(\cdot),q(\cdot)\big)$ is the amount of time
required for the trajectory $\big(n(\cdot), q(\cdot)\big)$ to reach and
subsequently remain in the set $\Jscr_\beps$.

\begin{theorem}\label{thm:attractive}
  For any $\beps > 0$, there exists $H_\beps > 0$ so that if $\big(n(\cdot),
  q(\cdot)\big)$ is a fluid trajectory of the \MWUMa policy in a critically
  loaded system, with initial condition satisfying $\|(n(0), q(0))\|_\infty
  \leq 1$, then
  \[
  h_\beps\big(n(\cdot), q(\cdot)\big) \leq H_\beps.
  \]
\end{theorem}

In order to prove Theorem~\ref{thm:attractive}, we require the following
technical lemma:
\begin{lemma}\label{lem:cont-radial}
  Under the \MWUMa policy, the lifting map $\Delta$ is continuous. Further,
  $\Delta$ is also
  positively homogeneous, i.e., for all $(n,q) \in \Rp^\Fscr\times\Rp^\Escr$
  and $\kappa > 0$,
  \[
  \Delta(\kappa n, \kappa q) = \kappa \Delta(n,q).
  \]
\end{lemma}
\begin{proof} 
  To establish continuity, it suffices to prove that is $(n^k, q^k) \to
  (n,q)$, then $\Delta(n^k, q^k) \to \Delta(n,q)$.
  By definition of $\Delta$, we have that $L_\alpha\big(\Delta(n,q)\big) \leq
  L_\alpha(n,q)$.  Therefore, since the convergent sequence $\{(n^k, q^k)\}$
  lies in a compact set, the sequence $\{\Delta(n^k, q^k)\}$ is contained in a
  compact set also. Define $x^k \defeq \Delta(n^k,q^k)$ and $x \defeq
  \Delta(n,q)$.  Then, there exists a convergence subsequence of $\{x^k\}$
  that converges to some $\hat{x}$. In what follows, that $\hat{x} = x$, and
  thereby complete the proof of continuity of $\Delta$.

  Suppose that $\hat{x}\neq x$. By passing to a subsequence, without loss of
  generality, assume that $x^k \to \hat{x}$. Since $x^k$ is feasible for
  $\ALGD(n^k,q^k)$, and $(n^k,q^k)\to (n,q)$, it follows that $\hat{x}$ is
  feasible for $\ALGD(n,q)$. Since $x$ is the unique optimal solution to
  $\ALGD(n,q)$, we have that $L_\alpha(x) < L_\alpha(\hat{x})$. 

  Now, define
  \[
  \beps^k \defeq \left(\max_{\zeta \in \CR^*(\rho)} \frac{w_\zeta(n^k, q^k)
      - w_\zeta(n,q)}{w_\zeta(\bone)}\right)^+,
  \]
  where $(x)^+ \defeq \max(0,x)$. Consider $\tilde{x}^k \defeq x + \beps^k
  \bone$. By the definition of $\beps^k$, it follows that $\tilde{x}^k$ is
  feasible for $\ALGD(n^k, q^k)$. Then, $L_\alpha(x^k) \leq
  L_\alpha(\tilde{x}^k)$. Now, as $k\to\infty$, $L_\alpha(x^k) \to
  L_\alpha(\hat{x})$. Further, $\beps^k \to 0$, so $\tilde{x}^k \to x$ and
  $L_\alpha(\tilde{x}^k) \to L_\alpha(x)$.  Then, $L_\alpha(\hat{x}) \leq
  L_\alpha(x)$. By contradiction, this establishes the continuity of $\Delta$.

  The positive homogeneity of $\Delta$ follows directly from the definition of
  the optimization problem $\ALGD$.
\end{proof}

\begin{proof}[Proof of Theorem~\ref{thm:attractive}]
Given $\delta > 0$, define
\begin{align*}
  \Dscr 
  & \defeq 
  \left\{ 
    (n, q)\in\Rp^\Fscr\times\Rp^\Escr\ 
    :\ 
    L_\alpha(n,q) \leq L_\alpha(\1)
  \right\},
  \\
  \Iscr 
  & \defeq
  \left\{ 
    (n, q)\in\Dscr\ 
    :\ 
    (n,q)=\Delta(n,q)
  \right\},
  \\
  \Iscr_\delta
  & \defeq
  \left\{ 
    (n, q)\in\Dscr\ 
    :\ 
    \| (n,q) - (n',q') \|_1 < \delta,\ (n',q')\in\Iscr
  \right\},
  \\
  \Kscr_\delta 
  & \defeq
  \left\{ 
    (n, q)\in\Dscr\ 
    :\ 
    K(n,q) < K(n',q'),\ \forall\ 
    (n',q')\in\Dscr\setminus\Iscr_\delta
  \right\}.
\end{align*}
where $K(n,q) \defeq L_\alpha(n,q) - L_\alpha\big(\Delta(n,q)\big)$. The
result can be established by showing that the following hold:
\begin{enumerate}[(i)]
  \itemsep=0pt
\item\label{enum:attr-i} $K\big(n(t), q(t)\big)$ is non-increasing in $t$. 

\item\label{enum:attr-ii} For $\delta > 0$ sufficiently small, $\Iscr \subset
  \Kscr_\delta \subset \Iscr_\delta \subset \Jscr_\beps$.

\item\label{enum:attr-iii} Starting from any initial condition in $\Dscr$
  (this includes all $(n,q)$ with $\|(n,q)\|_\infty \leq 1$), the time to hit
  $\Kscr_\delta$ is bounded uniformly.
\end{enumerate}

In particular, \eqref{enum:attr-iii} implies that starting from any state in
$\Dscr$, the fluid trajectory hits the set $\Kscr_\delta$ in finite time. By
\eqref{enum:attr-i}, once the trajectory is in set $\Kscr_\delta$, it remains
in that set forever. By \eqref{enum:attr-ii}, $\Kscr_\delta \subset
\Jscr_\beps$, and the result follows. To complete the proof,
\eqref{enum:attr-i}, \eqref{enum:attr-ii} and \eqref{enum:attr-iii} need to be
justified.

\medskip

\noindent
\emph{(\ref{enum:attr-i}):} 
Since $L_\alpha$ is a Lyapunov function, then $L_\alpha\big(n(t), q(t)\big)$ is
non-increasing over time under any fluid trajectory (cf.\
Lemma~\ref{lem:fluid-drain}). From Lemma~\ref{lem:wc}, the constraints in the
optimization problem $\ALGD\big(n(t), q(t)\big)$ can only become more
restrictive over time. Therefore, as the time $t$ increases, the cost of the
optimal solution of $\ALGD\big(n(t), q(t)\big)$, i.e.,
$L_\alpha\big(\Delta(n(t),q(t))\big)$, can only be non-decreasing.  Therefore,
$K\big(n(t), q(t)\big)$ is non-increasing over time.

\medskip

\noindent
\emph{(\ref{enum:attr-ii}):} First, consider claim $\Iscr \subset
\Kscr_\delta$ for any $\delta > 0$. $\Delta$ is continuous by
Lemma~\ref{lem:cont-radial}. The constraints, one for each $\zeta \in
\CR^*(\rho)$, in $\ALGD(n,q)$ are continuous with respect to $(n,q)$. And
$\big(n(t), q(t)\big)$ continuous with over $t$. Therefore, both functions
$L_\alpha\big(n(t), q(t)\big)$ and $K\big(n(t), q(t)\big)$ are continuous with
respect to $t$.  Now, $\Dscr$ is closed and bounded and $\Iscr_\delta$ is
open, hence $\Dscr \setminus \Iscr_\delta$ is closed and bounded. Therefore,
the infimum of the continuous function $K$ is achieved over this set. Since
$\Iscr \subset \Iscr_\delta$ for any $\delta > 0$, by the definition of
$\Iscr$, this this infimum must be strictly positive. However, over $\Iscr$
value of $K$ is $0$. Therefore, it follows that $\Iscr \subset
\Kscr_\delta$. The claim that $\Kscr_\delta \subset \Iscr_\delta$ is trivial
since if, $(n,q) \in \Dscr\setminus \Iscr_\delta$ then $K(n,q)$ is greater
than or equal to the infimum over that set, hence $(n,q)\notin
\Kscr_\delta$. Finally, to establish that $\Iscr_\delta \subset \Jscr_\beps$,
recall again that $\Delta$ is continuous and, hence, uniformly continuous over
$\Dscr$. Therefore, for any $\beps > 0$ there exists $\delta > 0$ such that
\[
\|(n,q) - (n',q')\|_1 < \delta 
\quad\Rightarrow\quad
\|\Delta(n,q) - \Delta(n',q')\|_1 < \beps/2.
\]
Consider any $(n,q) \in \Iscr_\delta$ and $(n',q') \in \Iscr$ with $\|(n,q) -
(n',q')\|_1 < \delta$. Then,
\[
\begin{split}
\|(n,q) - \Delta(n,q)\|_1 
& \leq 
\|(n,q) - (n',q')\|_1 + \|(n',q') - \Delta(n',q')\|_1 
+ \|\Delta(n',q') - \Delta(n,q)\|_1 \\
& \leq 
\delta + 0 + \beps/2 
\\
& < \beps, 
\end{split}
\]
for small enough choice of $\delta$. This completes the proof 
of (\ref{enum:attr-ii}). 

\medskip

\noindent
\emph{(\ref{enum:attr-iii}):} Here, we shall use
Theorem~\ref{thm:invariant}. First observe that, by the definition of $K$ and
Lemma~\ref{lem:wc}, for all $0 \leq s \leq t$,
\[
K\big(n(s),q(s)\big) - K\big(n(t),q(t)\big)
\geq
L_\alpha\big(n(s),q(s)\big) - L_\alpha\big(n(t),q(t)\big).
\]
In other words, the decrease in $K$ is at least as much as decrease in
$L$. 

Next, we wish to argue that when the fluid trajectory belongs to the set
$\Dscr \setminus \Kscr_\delta$ (i.e., is away from the space of fixed points
$\Iscr$) then $L_\alpha$ is strictly decreasing at some minimal rate. To be
precise, given $(n,q)\in \Dscr \setminus \Kscr_\delta$, suppose that
$\big(n(\cdot),q(\cdot)\big)$ is a fluid model solution and $t$ a regular
point such that $\big(n(t),q(t)\big)=(n,q)$. We would like to show that
\begin{equation}\label{eq:neg-D}
\mathfrak{D}(n,q) \defeq 
\frac{1}{1 + \alpha}\; 
\dt
L_\alpha\big(n(t),q(t)\big)
\leq - \gamma,
\end{equation}
for some $\gamma > 0$.

To this end, for each $f\in\Fscr$, define the function $x_f\colon
\Rp^\Fscr\times\Rp^\Escr\rightarrow [0,C]$ as
\[
x_f(n,q) \defeq 
\begin{cases}  
\rho_f & \text{if $n_f = 0$,} \\
n_{f}(t)/q_{\iota(f)}(t)
& \text{if $0 < n_f < C q_{\iota(f)}$,}
\\
C & \text{otherwise.}
\end{cases}
\]
Examining \FMSeq{enum:flow}, it is clear that this function determines the
rate allocated to $f$ at time $t$, i.e., $x_f(t) = x_f(n,q)$.
Recalling \eqref{eq:dtL}--\eqref{eq:B-bound}, we have that
\begin{equation}\label{eq:dtL-1}
  \begin{split}
    \fD(n,q)
    & =  \sum_{f \in \Fscr} 
    \frac{n_f^{\alpha}}{\mu_f \rho_f^\alpha}
    \dot{n}_f(t) 
    + 
    \sum_{e\in \Escr} 
    q_e^{\alpha} \dot{q}_e(t)  \\
    & = \sum_{f \in \Fscr} 
      \frac{n_f^{\alpha}}{\rho_f^\alpha}
      \big(\rho_f - x_f(n,q)\big)
    + 
    \left[ \Gamma x(n,q) \right]^\top q^\alpha
    -
    \max_{\sigma\in\Sscr}\ \sigma^\top (I - R) q^\alpha
    \\ 
    & =  \sum_{f\in \Fscr} 
    \left(
      \frac{n_{f}^{\alpha}}{\rho_{f}^\alpha}
      \big(\rho_{f} - x_{f}(n,q)\big) 
      + q_{\iota(f)}^\alpha \big(x_{f}(n,q) - \rho_{f}\big) 
    \right) 
    +
    \left[ \Gamma \rho \right]^\top q^\alpha
    -
    \max_{\sigma\in\Sscr}\ \sigma^\top (I - R) q^\alpha
    \\
    & = \sum_{f\in \Fscr}  \left(\frac{n_{f}^{\alpha}}{\rho_f^\alpha} 
      - q_{\iota(f)}^\alpha\right)
    \big(\rho_{f} - x_{f}(n,q)\big) 
    +
    T(q).
  \end{split}
\end{equation}
Here, for convenience, we define the function $T\colon \Rp^\Escr\to\R$ by
\[
T(q) \defeq
\left[ \Gamma \rho \right]^\top q^\alpha 
- \max_{\sigma\in\Sscr}\ \sigma^\top (I - R) q^\alpha. 
\]
Now, recall that $\rho_f < C$ for all $f \in \Fscr$. Therefore, it follows
that, for all $f\in \Fscr$, if $n_f > 0$,
\[ 
x_{f}(n,q) < \rho_{f} \quad \Longleftrightarrow \quad 
n_{f}/\rho_{f} < q_{\iota(f)}(t).
\]
Therefore, for all $f \in \Fscr$,
\begin{equation}
\label{eq:lx1}
\left(\frac{n_{f}^{\alpha}}{\rho_f^\alpha} - q_{\iota(f)}^\alpha\right) 
\big(\rho_{f} - x_{f}(n,q)\big)  \leq 0, 
\end{equation}
with the inequality being strict if $n_f \neq \rho_f q_{\iota(f)}$ and $n_f > 0$.

Since $(n, q) \in \Dscr \setminus \Kscr_\delta$, it can not be
a fixed point. Therefore, by part~(\ref{enum:lm-iv}) of the 
equivalence established in Theorem~\ref{thm:invariant}, one of 
the following two conditions holds:
\begin{itemize}
\item $T(q) < 0$.
\item $T(q) = 0$, and there exists some $f\in\Fscr$ with $n_f \neq \rho_f
  q_{\iota(f)}$ with $n_f > 0$. 
\end{itemize}
Note that the $n_f>0$ requirement of the second case follows from
\FMSeq{enum:nzero}: $n_f=0$ would imply that $q_{\iota(f)}=0$, hence if $n_f
\neq\rho_f q_{\iota(f)}$, it must be that $n_f \neq 0$.  In either of the
above two cases, using \eqref{eq:lx1} it is easy to see that the right hand
side of \eqref{eq:dtL-1} is strictly negative.  However, this does not provide
a uniform, strictly negative, bound on the drift $\fD$. If $\fD$ were
established to be a continuous function over set $\Dscr \setminus
\Kscr_\delta$, then such a uniform negative bound would follow as the set
$\Dscr \setminus \Kscr_\delta$ is closed and bounded.  However, closer
examination \eqref{eq:dtL-1} reveals that $\fD$ depends on the function $x_f$,
thus is not necessarily continuous at the boundary $n_f = 0$, for any $f \in
\Fscr$.

This difficulty is overcome as follows. We will cover $\Dscr \setminus
\Kscr_\delta$ by a finite collection of closed and bounded sets. On each set,
we will obtain a bound on the drift $\fD$ that is continuous on the set as
well as strictly negative. Hence, we will conclude that $\fD$ is uniformly
bounded by a negative quantity on $\Dscr \setminus \Kscr_\delta$, i.e., that
\eqref{eq:neg-D} holds. The details are given next.

To begin, define the function $\fR\colon\Rp\times\Rp\rightarrow [0,C]$ by
\[
\fR\big(n_f,q_{\iota(f)}\big) \defeq
\begin{cases}
n_f/q_{\iota(f)} & \text{if $n_f < Cq_{\iota(f)}$,}
\\
C & \text{otherwise.}
\end{cases}
\]
For given a vector $\bb = [\bb_f] \in \{-1,0,1\}^{|\Fscr|}$, define
$\Sbb$ as to be the set of $(n,q) \in \Dscr \setminus \Kscr_\delta$ such that, for each $f$,
\begin{align*}
  \rho_f - \fR\big(n_f,q_{\iota(f)}\big) & \geq n_f,
  & &
  \text{if $\bb_f=-1$},
  \\
  \left| \fR\big(n_f,q_{\iota(f)}\big) - \rho_f  \right| & \leq n_f,
  & &
  \text{if $\bb_f=0$},
  \\
  \fR\big(n_f,q_{\iota(f)}\big) - \rho_f  & \geq n_f,
  & &
  \text{if $\bb_f=1$}.
\end{align*}
Clearly $\Sbb$ is a closed and bounded set, since $\fR$ is continuous.  Given
$f\in\Fscr$ with $\bb_f\neq 0$, define the function $g_f\colon \Sbb\to\R$
by
\[
g_f\big(n,q) \defeq
\begin{cases}
  \rho_f^\alpha - 
  \left(\rho_f - n_f\right)^\alpha
  & \text{if $\bb_f=-1$,}
  \\
  \left(n_f  + \rho_f\right)^\alpha
  - \rho_f^\alpha
  & \text{if $\bb_f=1$.}
\end{cases}
\]
It is easy to check that for all $(n,q)\in\Sbb$,
\begin{equation}\label{eq:ra-bound}
\left|\fR^\alpha\big(n_f,q_{\iota(f)}\big) - \rho^\alpha_f\right|
\geq g_f(n,q) \geq 0.
\end{equation}
Further, $g_f$ is continuous and $g_f(n,q)=0$ if and only if $n_f=0$.
Finally, define the function $\fbb \colon \Sbb \to \R$ as 
\[
\begin{split}
\fbb(n,q) 
& \defeq T(q) 
+ 
\sum_{f\ :\ \bb_f = 0} 
\left(\frac{n_{f}^{\alpha}}{\rho_f^\alpha} 
  - q_{\iota(f)}^\alpha\right) \big(\rho_{f} - x_{f}(n,q)\big) 
\\
&
\quad
- 
\sum_{f\ :\ \bb_f \neq 0}  
\min\ \left( 
  (C-\rho_f)(\rho_f^{-\alpha}-C^{-\alpha}) n_f^{\alpha},
  \frac{n^{1+\alpha}_f g_f(n,q)}{C^\alpha \rho_f^\alpha}
\right).
\end{split}
\]

We make the following claims:
\begin{enumerate}[(a)]
\item\label{enum:iii-a} $\fbb$ is a continuous function over $\Sbb$.
\item\label{enum:iii-b} For any $(n,q) \in \Sbb$, $\fD(n,q) \leq \fbb(n,q)$.
\item\label{enum:iii-c} For any $(n,q) \in \Sbb$, $\fbb(n,q) < 0$.
\end{enumerate}

\medskip

\noindent
\emph{(\ref{enum:iii-a}):} To establish this claim, it is sufficient to
observe that for all $f$ with $\bb_f = 0$, $x_f$ is a continuous function over
$\Sbb$.  To see this, note that if $\bb_f = 0$, then
$\big|\fR\big(n_f,q_{\iota(f)}\big) - \rho_f\big| \leq n_f$. Now $x_f(n, q) =
\fR\big(n_f,q_{\iota(f)}\big)$ if $n_f > 0$ and $x_f(n, q) = \rho_f$ if $n_f =
0$. Therefore, over $\Sbb$, we have that $x_f(n,q) =
\fR\big(n_f,q_{\iota(f)}\big)$ for all $n_f \geq 0$.  This establishes
continuity of $x_f$ for $f$ with $\bb_f = 0$.

\medskip

\noindent
\emph{(\ref{enum:iii-b}):} Here, we need to show that for any $f$ with $\bb_f
\neq 0$, the term in $\fbb$ is larger than or equal to the corresponding term
on the right hand side of \eqref{eq:dtL-1} in magnitude and preserves the
sign. That is,
\begin{equation}\label{eq:iii-mag}
\left|\frac{n_{f}^{\alpha}}{\rho_f^\alpha} - q_{\iota(f)}^\alpha \right| 
\big|\rho_{f} - x_{f}(n,q)\big| 
\geq 
\min\ \left( 
  (C-\rho_f)(\rho_f^{-\alpha}-C^{-\alpha}) n_f^{\alpha},
  \frac{n^{1+\alpha}_f  g_f(n,q)}{C^\alpha \rho_f^\alpha}
\right),
\end{equation}
and 
\begin{equation}\label{eq:iii-sign}
\left(\frac{n_{f}^{\alpha}}{\rho_f^\alpha} - q_{\iota(f)}^\alpha \right) 
\big(\rho_{f} - x_{f}(n,q)\big) = 0 
\quad \Longleftrightarrow \quad
n_f = 0.
\end{equation}

Now, if $n_f = 0$ then $x_f(n,q) = \rho_f$ and hence the left hand side above
of \eqref{eq:iii-sign} is $0$. If $n_f > 0$, since $\bb_f \neq 0$ and thus
$\big|\fR\big(n_f,q_{\iota(f)}\big) - \rho_f\big| \geq n_f > 0$, we have $x_f
\neq \rho_f$.  Therefore, the left hand side of \eqref{eq:iii-sign} is not
equal to $0$. Thus, \eqref{eq:iii-sign} is established. 

To prove \eqref{eq:iii-mag}, we have the following cases:
\begin{itemize}
\item $n_f=0$.
  Here, both sides of \eqref{eq:iii-mag} are $0$.

\item $0 < n_f \leq C q_{\iota(f)}$. Here, we have
  \[
  \begin{split}
    \left|\frac{n_{f}^{\alpha}}{\rho_f^\alpha} - q_{\iota(f)}^\alpha \right| 
    \big|\rho_{f} - x_{f}(n,q)\big| 
    & 
    =
    \frac{q_{\iota(f)}^\alpha}{\rho_f^\alpha}
    \left|\frac{n_{f}^{\alpha}}{q_{\iota(f)}^\alpha} - \rho_f^\alpha \right| 
    \big|\rho_{f} - x_{f}(n,q)\big| 
    \\
    &
    =
    \frac{q_{\iota(f)}^\alpha}{\rho_f^\alpha}
    \left|\fR^\alpha\big(n_f,q_{\iota(f)}\big) - \rho_f^\alpha \right| 
    \big|\rho_{f} - \fR\big(n_f,q_{\iota(f)}\big)\big| 
    \\
    &
    \geq
    \frac{n_f q_{\iota(f)}^\alpha g_f(n,q)}{\rho_f^\alpha}
    \geq
    \frac{n^{1+\alpha}_f g_f(n,q)}{C^\alpha \rho_f^\alpha},
  \end{split}
  \]
  where we have used the fact that $\bb_f\neq 0$ and \eqref{eq:ra-bound}.
\item $0 \leq C q_{\iota(f)} < n_f$. Here, since $\rho_f < C$, we have that
  \[
  \begin{split}
    \left|\frac{n_{f}^{\alpha}}{\rho_f^\alpha} - q_{\iota(f)}^\alpha \right| 
    \big|\rho_{f} - x_{f}(n,q)\big| 
    & 
    =
    \left(\frac{n_{f}^{\alpha}}{\rho_f^\alpha} - q_{\iota(f)}^\alpha \right)
    (C - \rho_{f})
    \geq
    (C-\rho_f)(\rho_f^{-\alpha}-C^{-\alpha}) n_f^{\alpha}.
  \end{split}
  \]
\end{itemize}

\medskip

\noindent
\emph{(\ref{enum:iii-c}):} Suppose that $(n,q) \in \Sbb$. We wish to establish
that $\fbb(n,q) < 0$. Since $(n,q)$ is not an invariant point, by
part~(\ref{enum:lm-iv}) of the equivalence established in
Theorem~\ref{thm:invariant} and by \FMSeq{enum:nzero}, one of the following
two conditions holds:
\begin{itemize}
\item $T(q) < 0$. In this case, using \eqref{eq:lx1}, clearly $\fbb(n,q) < 0$.
\item $T(q) = 0$, and there exists some $f\in\Fscr$ with $n_f \neq \rho_f
  q_{\iota(f)}$ and $n_f > 0$. Here, if $\bb_f=0$, then since the inequality in
  \eqref{eq:lx1} must be strict, we have $\fbb(n,q) < 0$. On the other hand,
  suppose that $\bb_f\neq 0$.  Since $n_f > 0$, we have that
  \[
\min\ \left( 
  (C-\rho_f)(\rho_f^{-\alpha}-C^{-\alpha}) n_f^{\alpha},
  \frac{n^{1+\alpha}_f  g_f(n,q)}{C^\alpha \rho_f^\alpha}
\right)
  > 0,
  \]
  and it follows that $\fbb(n,q) < 0$. 
\end{itemize}

\medskip

Now, given claims (\ref{enum:iii-a}) and (\ref{enum:iii-c}), it follows that
\[
\sup_{(n,q)\in\Sbb}
\fbb(n,q) \leq -\gamma_{\bb} < 0,
\]
for some $\gamma_\bb > 0$. Using claim (\ref{enum:iii-b}), we have that
\[
\fD(n,q) \leq -\gamma_{\bb} < 0,
\]
for all $(n,q)\in\Sbb$. Since $\Dscr \setminus \Kscr_\delta = \cup_{\bb \in
  \{-1,0,1\}^{\Fscr}} \Sbb$, we have that, for all $(n,q)\in \Dscr\setminus
\Kscr_\delta$,
\[
\fD(n,q) \leq -\gamma < 0,
\]
where
\[
\gamma \defeq \min_{\bb\in \{-1,0,1\}^\Fscr}\ \gamma_\bb.
\]
This completes the proof of Theorem \ref{thm:attractive}.

\end{proof}

\subsection{Fixed Points: Optimality}

The following theorem characterizes the cost associated with an invariant
state, relative to the effective cost. The effective cost represents the
lowest cost achievable under \emph{any} policy (cf.\
Theorem~\ref{thm:eff-cost}). Hence, this result implies that the invariant
states of the \MWUMa policy are cost optimal, as $\alpha\to 0^+$.

\begin{theorem}
  Suppose $(n^*,q^*)$ is an invariant state of a critically loaded system
  under the \MWUMa policy. Then,
  \begin{equation}\label{eq:inv-cost}
    c(n^*,q^*) \leq \big(1 + \beta(\alpha)\big)c^*(n^*,q^*),
  \end{equation}
  where $\beta(\alpha)\to 0$ as $\alpha\to 0^+$.
\end{theorem}
\begin{proof}
  Suppose $(n^*,q^*)$ is an invariant state.
  Define $(n',q')$ to be an
  optimal solution to the effective cost LP $c^*(n^*,q^*)$, defined by
  \eqref{eq:eff-cost}.
  Clearly
  \[
  L_\alpha(n^*,q^*) \leq L_\alpha(n',q'),
  \]
  since $(n^*,q^*)$ is optimal
  for $\ALGD(n^*,q^*)$, and $(n',q')$ is feasible for $\ALGD(n^*,q^*)$. Then,
  following the same argument as in Lemma~\ref{lem:ub},
  \[
  c(n^*,q^*) 
  \leq \big(1+\beta(\alpha)\big) c(n',q')
  = \big(1+\beta(\alpha)\big) c^*(n^*,q^*),
  \]
  where $\beta(\alpha)\to 0$ as $\alpha\to 0^+$.
\end{proof}

\section{Discussion and Future Work}\label{sec:conclusion}

We have provided a model of a communications network that operates at the
packet-level with the goal of achieving end-to-end performance at the
flow-level. The proposed \MWUMa control policy achieves this goal by means of
the maximum weight-$\alpha$ packet-level scheduling along with the
$\alpha$-fair rate allocation. We established the positive recurrence of the
system by means of fluid model when the system is underloaded.  For the
critically loaded fluid model, we established path-wise constant factor
optimality; the constant factor depends $\alpha$ and the {\em balance factor}.

There are several interesting directions for future work. To start with, by
characterizing the invariant manifold of the critically loaded fluid model and
establishing its attractiveness, the work here should lead to the
multiplicative state-space collapse property in a relatively straightforward
manner following the method of Bramson \cite{Bramson98}. As the next step,
establishing the strong state-space collapse property would require bounding
the the maximal deviation in the system state over certain time-horizon. We
strongly believe that under \MWUMa control policy for $\alpha \geq 1$, this
should follow from a recently developed Lyapunov function based maximal
inequality by Shah, Tsitsiklis and Zhong \cite{STZ10}. However, further
obtaining a complete characterization of the diffusion (heavy traffic)
approximation seems to be far more non-trivial question. Finally, the results
about path-wise constant factor optimality of critically loaded fluid model
seem to suggest the possibility of such constant factor optimality of \MWUMa
control policy under diffusion approximation.

\bibliographystyle{plain}
\bibliography{biblio}

\clearpage

\appendix

\section{Standard Norm Inequalities}

The following lemma provides some standard norm inequalities that are used
throughout the paper:

\begin{lemma}\label{lem:norm}
  Consider a vector $x\in\R^d$. 
  \begin{enumerate}[(i)]
  \item\label{enum:p-q} 
    If $0 <  p \leq q \leq \infty$, then $\|x\|_q \leq \|x\|_p$.
  \item\label{enum:one-p}
    If $1 \leq p \leq \infty$ and $1/p + 1/q = 1$, then
    \[
    d^{-1/q} \|x\|_1 \leq \|x\|_{p}.
    \]
  \end{enumerate}
\end{lemma}

\section{Justification of the Fluid Model}\label{ap:fluid}

In this appendix, we will provide a proof for Theorem~\ref{thm:fms}, which
establishes fluid model as the formal functional law of large numbers 
approximation. 

We begin with some technical preliminaries. Fix $T > 0$. Recall from
Section~\ref{sec:formal-fms} that $\bD[0,T]$ is the space of functions from
$[0,T]$ to $\mathfrak{Z}$, as defined in \eqref{eq:fms-path}, that are
RCLL. This space is equipped with the Skorohod metric, defined as
\[
\bd(\bdx, \by) \defeq \inf_{\phi \in \Phi}\ \|\phi\|^o \vee \|\bdx - \phi\circ \by\|,\quad\text{for }
\bdx, \by \in \bD[0,T].
\]
Here, $\Phi$ is the set of all non-decreasing functions $\phi\colon [0,T] \to
[0,T]$ with $\phi(0)=0$ and $\phi(T)=T$.  The norm $\|\cdot\|^o$ over $\Phi$
is defined as follows: for $\phi \in \Phi$,
\[
\|\phi\|^o \defeq\sup_{0\leq s < t \leq T}\ \log \left| \frac{\phi(t)-\phi(s)}{t-s}\right|.
\]
By $\phi \circ \by$ refers to the composition $\by(\phi(t))$, and for 
any $\bdx \in \bD[0,T]$, 
\[
\|\bdx\| \defeq \sup_{t\in [0, T]} \|\bdx(t)\|_1,
\]
with $\|\cdot\|_1$ being the standard $\ell_1$-norm over the product space
$\mathfrak{Z}$.

For any $\bdx \in
\bD[0,T]$ and $0\leq s < t \leq T$, define
\[
 w_\bdx(s,t) \defeq \sup_{s\leq t_1, t_2\leq t}\ \|\bdx(t_1) - \bdx(t_2)\|_1.
\]
Further, for any $\delta > 0$, define
\[
w_\bdx^\prime(\delta) = \inf_{\{t_i\} \in \bS(T,\delta)} \max_i\ w_\bdx(t_{i-1},t_i),
\]
where $\bS(T,\delta)$ is collection of all $\delta$-sparse decompositions
$\{t_i\}$ of $[0,T]$, i.e., $0 = t_0 < t_1 < \dots < t_\ell = T$ with $t_i -
t_{i-1} \geq \delta$ for all $i\geq 1$.

It can be easily checked (see \cite[Chapter~3]{Billingsley99}) that the
$\bD[0,T]$ is Polish space under the metric $\bd$. Let $\Bscr_T$ denote the
Borel $\sigma$-algebra on $\bD[0,T]$ with respect to the topology induced by
$\bd$.  We will be interested in probability measures over space $(\bD[0,T],
\Bscr)$.  We shall utilize the following well-known characterization of
tightness of measures (see \cite[Theorem~13.2]{Billingsley99}):
\begin{theorem}\label{thm:tight}
  The collection of measures $\{P_\theta\ :\ \theta \in \Theta\}$ defined on
  $(\bD[0,T], \Bscr)$ is tight if and only if the following two conditions are
  satisfied:
  \begin{enumerate}[(a)]
  \item\label{enum:tight-a}
    \[
    \lim_{A\to\infty} \limsup_{\theta \in \Theta}\ 
    P_\theta\Big( \bdx \in \bD[0,T]\ :\ \|\bdx\| \geq A \Big) = 0.
    \]
 
  \item\label{enum:tight-b}
    For each $\varepsilon > 0$, 
    \[
    \lim_{\delta \to 0} \limsup_{\theta\in\Theta}\
    P_\theta\Bigl( \bdx \in \bD[0,T]\ :\
    w_\bdx^\prime(\delta) \geq \varepsilon\Bigr) = 0.
    \]
  \end{enumerate}
\end{theorem}

We state the following well-known `concentration' property of Poisson process
that shall later be useful.  It follows from the application of a standard
Chernoff bound (see, for example, \cite[Theorem~2.2.3]{DZ03}).
\begin{proposition}\label{prop:Poi}
Consider a Poisson process of rate $1$. Let $N(t)$ be
the number of events of this Poisson process in
time interval $[0,t]$. Then, for any $\delta\in [0,t]$,
\[
\PR\Bigl(\left| N(t) - t \right| \geq \delta \Bigr)
\leq 2\exp\left(-\frac{\delta^2}{2t}\right).
\]
\end{proposition}

\subsection{Tightness}


The first step in the proof of Theorem~\ref{thm:fms} is the following lemma,
which establishes tightness of the collection of measures associated with the
scaled system processes.

\begin{lemma}\label{lem:mu-tight}
  Under the hypotheses of Theorem~\ref{thm:fms}, for each $r \geq 1$, let
  $\mu^{(r)}$ denote the measure of $Z^{(r)}(\cdot) \in \bD[0,T]$. Then, the
  collection of measures $\{\mu^{(r)}\ :\ r \geq 1\}$ is tight.
\end{lemma}
\begin{proof}
  We will establish tightness by verifying conditions (\ref{enum:tight-a}) and
  (\ref{enum:tight-b}) of Theorem~\ref{thm:tight}.

  First, consider condition (\ref{enum:tight-a}). It is sufficient
  to show that for any $\delta > 0$, there exist $K(\delta)$ and $r(\delta)$
  such that, for $K \geq K(\delta)$ and $r \geq r(\delta)$,
  \begin{equation}\label{eq:t1}
    \mu^{(r)}\Bigl(\bdx \in \bD[0,T]\ :\ \|\bdx\| \geq K\Bigr) \leq \delta.
  \end{equation}
  To establish this, fix $\delta > 0$. By definition,
  \[
  \begin{split}
    \| \Zscr^{(r)}(\cdot)\|
    & = \sup_{t\in [0,T]}\ \Big(\|Q^{(r)}(t)\|_1 + 
    \|Z^{(r)}(t)\|_1 + \|N^{(r)}(t)\|_1 + \|S^{(r)}(t)\|_1 + \|\bar{X}^{(r)}(t)\|_1  \\ 
    & \qquad\qquad\qquad + \|\An^{(r)}(t)\|_1 + \|D^{(r)}(t)\|_1 + \|A^{(r)}(t)\|_1
    \Big).
  \end{split}
  \]
  We will bound each component of the system process $\Zscr^{(r)}(\cdot)$.

  First, observe that for any $t \in [0,T]$, with probability $1$,
  \begin{equation}\label{eq:t3}
    \|Z^{(r)}(t)\|_1 + \|S^{(r)}(t)\|_1 + \|\bar{X}^{(r)}(t)\|_1
    \leq K_1 T,
  \end{equation}
  where the constant $K_1$ depends on system dimensions $|\Escr|$ and
  $|\Fscr|$ and on the maximum rate allocation $C$. This is because at most a
  unit amount of scheduling can be performed per unit time, and maximal rate
  allocated to any flow type is at most $C$. 

  Next, consider the term $\|A^{(r)}(t)\|_1$.  For each flow type $f \in
  \Fscr$, $A^{(r)}_f(T)$ is a Poisson process with a time-varying rate that is
  at most $C$. Therefore, $A^{(r)}_f(T)$ is bounded above by $1/r$ times the
  total number of events of a Poisson process of rate $C$ in time interval
  $[0,rT]$.  This number of events is distributionally equivalent to number of events of
  Poisson process of rate $1$ in interval $[0,rTC]$. Therefore, using
  Proposition~\ref{prop:Poi}, 
  \begin{equation}\label{eq:t4-a}
    \PR\left( A^{(r)}_f(T) \geq 2CT \right)  
    \leq 
    2 e^{-\tfrac{1}{2} r T C}.
  \end{equation}
  It follows that for $r$ sufficiently large, for any $f\in\Fscr$,
  \begin{equation}\label{eq:t4}
    \PR\left( A^{(r)}_f(T) \geq 2CT \right)
    \leq  \frac{\delta}{10 |\Fscr|}.
  \end{equation}
  Then, by the union bound,
  \begin{equation}\label{eq:t5}
    \PR\left( \|A^{(r)}(T)\|_1 \geq 2 |\Fscr|CT \right) 
    \leq  \frac{\delta}{10}.
  \end{equation}
  
  Next, consider term $\|Q^{(r)}(t)\|_1$. Note that
  \begin{equation}\label{eq:t6}
    \sup_{t\in [0,T]}\ \|Q^{(r)}(t)\|_1 
    \leq
    \|Q^{(r)}(0)\|_1 + \|A^{(r)}(T)\|_1.
  \end{equation}
  Now, by hypothesis, we have $Q^{(r)}(0)\to q(0)$ with probability $1$ as
  $r\to\infty$. Therefore, the collection of vectors $\{Q^{(r)}(0)\}$ is
  almost surely bounded, and thus there exists a constant $K_2$ so that, for
  $r$ sufficiently large,
  \begin{equation}\label{eq:t7}
    \PR\left(\|Q^{(r)}(0)\|_1\geq K_2 \right) \leq \frac{\delta}{10}.
  \end{equation}
  It follows that for a large enough constant $K_3$ (dependent on $K_2$, $C$,
  $T$, $|\Escr|$, $|\Fscr|$), and for $r$ sufficiently large,
  \begin{equation} \label{eq:t9}
    \PR\left(\sup_{t\in[0,T]}\ \|Q^{(r)}(t)\|_1 \geq K_3 \right) 
    \leq  \frac{\delta}{5}.
  \end{equation}

  Using very similar arguments (employing Proposition~\ref{prop:Poi}), it is
  possible to bound the Poisson processes $\|D^{(r)}(t)\|_1$ and
  $\|\An^{(r)}(t)\|_1$. This, in turn, will a lead to a bound on
  $\|N^{(r)}(t)\|_1$, since
  \[
  \sum_{t\in [0,T]}\ \|N^{(r)}(t)\|_1 \leq \|N^{(r)}(0)\|_1 + \|\An^{(r)}(T)\|_1.
  \]
  Therefore, there exists a constant $K_4$, so that for $r$ sufficiently large,
  \begin{equation} \label{eq:t10}
    \PR\left(\sup_{t\in[0,T]}\ 
      \|D^{(r)}(t)\|_1 + \|\An^{(r)}(t)\|_1 + \|N^{(r)}(t)\|_1
      \geq K_4 \right) 
    \leq  \frac{\delta}{10}.
  \end{equation}

  From the discussion above, equations \eqref{eq:t3}, \eqref{eq:t5},
  \eqref{eq:t9}, \eqref{eq:t10}, and union bound, it follows that for any
  $\delta > 0$, there exists constants $K(\delta)$ and $r(\delta)$ such that
  for $K \geq K(\delta)$ and $r\geq r(\delta)$, we have that
  \[
  \PR\left(\|\Zscr^{(r)}(\cdot)\| \geq K\right)  \leq \delta.
  \]
  This completes the verification of \eqref{eq:t1} or equivalently, condition
  (\ref{enum:tight-a}) of Theorem~\ref{thm:tight}.

  Next, consider condition (\ref{enum:tight-b}) of
  Theorem~\ref{thm:tight}. For this, it is sufficient to show that for any
  $\varepsilon > 0$, there exist $\delta(\varepsilon)$ and $r(\delta(\varepsilon))$ 
  so that for $r \geq r(\delta(\varepsilon))$,
  \[
  \PR\left(w_{\Zscr^{(r)}}^\prime(\delta(\varepsilon)) \geq \varepsilon \right) \leq \delta(\varepsilon),
  \]
  with $\delta(\varepsilon) \to 0$ as $\varepsilon \to 0$. 
  
  To bound $w_{\Zscr^{(r)}}^\prime(\delta)$, we need to find an appropriate
  $\delta$-sparse decomposition $\{t_0,t_1,\ldots,t_n\}$ of $[0,T]$. For this,
  we consider a natural decomposition: $t_i = i\delta$, for $0 \leq i < n$,
  and $t_n = T$, when $n = \lfloor T/\delta \rfloor$.  Then, it follows that
  $\delta \leq t_i - t_{i-1} \leq 2\delta$ for all $0 < i \leq n$.

  We wish to bound $\|\Zscr^{(r)}(s) - \Zscr^{(r)}(t)\|_1$, for $t_{i-1}
  \leq s, t \leq t_i$, for any $0 < i \leq n$.  To this end, first note that
  \begin{equation}\label{eq:t12}
    \begin{split}
      \| \Zscr^{(r)}(s)-\Zscr^{(r)}(t)\|_1 & =  
      \|Q^{(r)}(t)-Q^{(r)}(s)\|_1 
      + \|Z^{(r)}(t)-Z^{(r)}(s)\|_1 
      + \|N^{(r)}(t)-N^{(r)}(s)\|_1   
      \\
      & \quad
      + \|S^{(r)}(t) - S^{(r)}(s)\|_1
      + \|\bar{X}^{(r)}(t) - \bar{X}^{(r)}(s)\|_1 
      + \|\An^{(r)}(t) - \An^{(r)}(s)\|_1  
      \\
      & \quad
      + \|D^{(r)}(t) - D^{(r)}(s)\|_1 
      + \|A^{(r)}(t) - A^{(r)}(s)\|_1.
    \end{split}
  \end{equation}

  As noted earlier, the terms involving $Z^{(r)}(\cdot)$, $S^{(r)}(\cdot)$ and
  $\bar{X}^{(r)}(\cdot)$ are collectively upper bounded by $K_1 |t-s|$, with a
  system dependent constant $K_1$, since they are all Lipschitz
  continuous. Therefore, for $\delta \leq \varepsilon/(10 K_1)$, the sum of
  these terms in \eqref{eq:t12} is no more than $\varepsilon/10$ with
  probability $1$.  Next, we consider the remaining five terms in
  \eqref{eq:t12}.  As in the justification of condition (\ref{enum:tight-a}),
  we will have similar argument for all these of five terms. We will focus on
  the term corresponding to $Q^{(r)}(\cdot)$. From \eqref{eq:prim-q}, it
  follows that
  \[
  \begin{split}
    \|Q^{(r)}(t)-Q^{(r)}(s)\|_1 & \leq 
    \big\|(I-R^\top) \Pi \big(S^{(r)}(t)-S^{(r)}(s)\big)\big\|_1 
    + \big\|(I-R^\top) \big(Z^{(r)}(t)-Z^{(r)}(s)\big)\big\|_1
    \\
    & \quad + \|A^{(r)}(t)-A^{(r)}(s)\|_1. 
  \end{split}
  \]
  Now the first two terms are bounded by $K_5 | t-s|$, with the constant $K_5$
  dependent on $|\Sscr|$, $|\Escr|$, and $|\Fscr|$.  To see this, note that
  both $S^{(r)}(\cdot)$ and $Z^{(r)}(\cdot)$ are Lipschitz continuous, and the
  matrices $(I-R^\top)\Pi$ and $(I-R^\top)$ are finite dimensional (with
  dimension dependent on $|\Sscr|$, $|\Escr|$, and $|\Fscr|$), and with each
  entry bounded by constants. It follows that by choosing $\delta \leq
  \varepsilon/(10 K_5)$, we have that
  \[
    \big\|(I-R^\top) \Pi \big(S^{(r)}(t)-S^{(r)}(s)\big)\big\|_1 
    + \big\|(I-R^\top) \big(Z^{(r)}(t)-Z^{(r)}(s)\big)\big\|_1
    \leq
    \frac{\varepsilon}{10},
  \]
  with probability $1$.  Now, to bound the contribution of term
  $\|A^{(r)}(t)-A^{(r)}(s)\|_1$, we can utilize arguments used for obtaining
  \eqref{eq:t5} with $T$ replaced by $|t-s|$. Note that here we need to use
  `memory-less' property of the Poisson process crucially. As a conclusion, we
  obtain that there exists a constant $K_6$ so that is $\delta \leq
  \varepsilon/(10 K_6)$, then for sufficiently large $r$,
  \begin{equation}\label{eq:t13}
    \PR\left(\|A^{(r)}(t)-A^{(r)}(s)\|_1
      \leq \varepsilon/10\right) \geq 1-\frac{\delta^2}{10 T}.
  \end{equation}

  From the above discussion, it follows that for any interval $[t_{i-1},t_i]$,
  as long as we choose $\delta$ sufficiently small and $r$ sufficiently large,
  \begin{equation}\label{eq:t14}
    \PR\left(\sup_{t_{i-1}\leq s, t \leq t_i}\ 
      \|Q^{(r)}(t) - Q^{(r)}(s)\|_1 \leq \varepsilon/5\right) \geq  
    1-\frac{\delta^2}{10 T}. 
  \end{equation}
  In a very similar manner (using Proposition~\ref{prop:Poi}), we obtain the
  following: there exists a constant $K_7$, so that if $\delta = 
  \varepsilon/K_7$ and $r$ is sufficiently large, for any $0 < i \leq n$,
  \begin{align}
    \PR\left(
      \sup_{t_{i-1}\leq s, t \leq t_i} \|N^{(r)}(t) - N^{(r)}(s)\|_1 
      \leq \varepsilon/10\right) & \geq 1-\frac{\delta^2}{10 T},
    \label{eq:t15} 
    \\
    \PR\left(
      \sup_{t_{i-1}\leq s, t \leq t_i} \|\An^{(r)}(t) - \An^{(r)}(s)\|_1 
      \leq \varepsilon/10\right) & \geq 1-\frac{\delta^2}{10 T}, 
    \label{eq:t16} 
    \\
    \PR\left(
      \sup_{t_{i-1}\leq s, t \leq t_i} \|A^{(r)}(t) - A^{(r)}(s)\|_1 
      \leq \varepsilon/10\right) & \geq 1-\frac{\delta^2}{10 T}, 
    \label{eq:t17} 
    \\
    \PR\left(
      \sup_{t_{i-1}\leq s, t \leq t_i} \|D^{(r)}(t) - D^{(r)}(s)\|_1 
      \leq \varepsilon/10\right) & \geq 1-\frac{\delta^2}{10 T}. 
    \label{eq:t18}
  \end{align}

  From \eqref{eq:t12}-\eqref{eq:t18} and the discussion above, it follows that
  there exists a constant $K'$ so that is $\varepsilon \leq K$ and $r$ is
  sufficiently large, by a union bound over at most $T/\delta$ intervals in
  the partition, we obtain that
  \[
  \PR\left(
    \max_{i} \sup_{t_{i-1}\leq s, t \leq t_i}\
    \|\Zscr^{(r)}(t) - \Zscr^{(r)}(s)\|_1
    \leq \varepsilon\right)  \geq  1-\delta.
  \]
  The result follows.
\end{proof}

\subsection{Proof of Theorem~\ref{thm:fms}}

We are now ready to prove Theorem~\ref{thm:fms}.

Given that the collection of measures $\{\mu^{(r)}\ :\ r \geq 1\}$ is tight,
for any sequence $\{r_k\ :\ k\in\N \}\subset\R$ with $r_k \to \infty$ as
$k\to\infty$, there exists a further subsequence $\{r_{k_\ell}\}$ and limit
point $\mu^{(\infty)}$ such that $\mu^{(r_{k_\ell})}$ converges weakly to
$\mu^{(\infty)}$ as $\ell \to\infty$.  By restricting to this subsequence,
assume that $\mu^{(r_k)} \Rightarrow \mu^{(\infty)}$. Since these measures are
defined on a Polish space, by the Skorohod representation theorem, there
exists a probability space over which we can define, for all $k$, random
variables $\Zscr^{(r_k)}(\cdot)$ and $\zscr(\cdot)$ that are distributed
according to $\mu^{(r_k)}$ and $\mu^{(\infty)}$, respectively, and where the
$\Zscr^{(r_k)}(\cdot)$ almost surely converges to $\zscr(\cdot)$, in the
Skorohod metric. We will use this setting to argue that the limiting random
variable $\zscr(\cdot)$ satisfies the appropriate fluid model equations with
probability $1$.  Subsequently, we will establish that under an arbitrary
control policy, $\mu^{(\infty)}\big(\FMS(T)\big) = 1$, while under the \MWUMa
policy, $\mu^{(\infty)}\big(\FMS^\alpha(T)\big) = 1$.  Since $\mu^{(r_k)}
\Rightarrow \mu^{(\infty)}$, from definition of weak convergence, it follows
that, under an arbitrary control policy, $\liminf_{k\to \infty}
\mu^{(r_k)}(\FMS_\varepsilon(T)) = 1$, and under the \MWUMa policy,
$\liminf_{k\to \infty} \mu^{(r_k)}(\FMS^\alpha_\varepsilon(T)) = 1$, for any
$\varepsilon > 0$. This will imply the desired result.

To this end, we start by establishing that $\mu^{(\infty)}\big(\FMS(T)\big) =
1$. That is, we need to show that equations
\FMSeq{enum:one}-\FMSeq{enum:eight} are satisfied.

\medskip

\noindent{\em \FMSeq{enum:one}, \FMSeq{enum:two}, \FMSeq{enum:five}:}
We start with \FMSeq{enum:one}.  Among the components of
$\Zscr^{(r_k)}(\cdot)$, $Z^{(r_k)}(\cdot), S^{(r_k)}(\cdot)$ and
$\bar{X}^{(r_k)}(\cdot)$ are Lipschitz continuous by construction over
$[0,T]$. Since $\Zscr^{(r_k)}(\cdot)$ converges almost surely to
$\zscr(\cdot)$ with respect to the Skorohod metric $\bd$, it follows that the
corresponding components of $\zscr(\cdot)$, $z(\cdot), s(\cdot)$ and
$\bar{x}(\cdot)$, are Lipschitz continuous. Equivalently, this follows by
the Arzel\`{a}-Ascoli theorem.

Now, to establish Lipschitz continuity of the other components of $\zscr$ we
will use the following result:

\begin{lemma}\label{lem:PoiLip}
  Given a fixed $T > 0$, consider a Poisson process $\Pscr$ with time-varying
  rate given by $\gamma(t)$, for $t \geq 0$. Assume there exists constant $K >
  0$ such that $\gamma(t) \in [0,K]$ for all $t \geq 0$ and that $\gamma(t)$
  depends only on events that happen up to time $t$. Consider a sequence
  $\{\theta_i\}\subset\Rp$ with $\theta_i \to \infty$ as
  $i\to\infty$.  Define the scaled process
  \[
  \Pscr^i(t) = \frac{1}{\theta_i} \Pscr(\theta_i t),
  \]
  for $t \in [0,T]$. Also, define the processes
  \[
  \gamma^i(t) = \frac{1}{\theta_i} \gamma(\theta_i t),
  \quad \text{and} \quad 
  \bar{\gamma}^i(t) = \int_{0}^t \gamma^i(s)\, ds,
  \]
  for $t\in [0,T]$, where we assume $\bar{\gamma}^i(\cdot)$ is well-defined.
  Assume that $\Pscr^i(\cdot)$ converges weakly to $\Pscr^\infty(\cdot)$.
  Then, the sample paths (over $[0,T]$) of $\Pscr^\infty(\cdot)$ are Lipschitz
  continuous with probability $1$. Further, assume that
  $\bar{\gamma}^i(\cdot)$ converges (u.o.c.) to $\bar{\gamma}(\cdot)$ over
  $[0,T]$. Then, $\Pscr^\infty(t) = \bar{\gamma}(t)$, for all $t\in [0,T]$.
\end{lemma}
\begin{proof}
  This is a well-known property of Poisson processes. We describe the key
  steps of the proof. First, using the concentration property of
  Proposition~\ref{prop:Poi}, it can be established that the sample paths of
  scaled Poisson process are approximately Lipschitz.  Second, using a
  variation of the Arzel\`{a}-Ascoli theorem (e.g., see
  \cite[Lemma~4.2]{Bramson98} or \cite[Lemma~6.3]{YOY05}), it can be
  established that the limit points of such approximately Lipschitz sample
  paths are in fact Lipschitz. Finally, note the fact that number of events
  under a Poisson process with time-varying rate $\gamma(\cdot)$ over time
  interval $[0,t]$ is distributionally equivalent to number of events under a
  unit rate Poisson process over time interval $[0, \int_0^t \gamma(s)\,
  ds]$. As long as $\gamma(\cdot)$ is uniformly bounded by some constant, say
  $K$, the functional strong law of large numbers for scaled unit rate Poisson
  process over $[0, KT]$ can be used to obtain the final desired claim. An
  interested reader may find the details to this argument in many places in
  literature (e.g., \cite[Appendix~A]{Massoulie07}).
\end{proof}
Now consider components $\An^{(r_k)}(\cdot)$, $A^{(r_k)}(\cdot)$, and
$D^{(r_k)}(\cdot)$.  By their construction, these are Poisson processes with
possibly time-varying rates that are always uniformly bounded. These processes
converge (over $[0,T]$) to the corresponding components $\an(\cdot), a(\cdot)$
and $d(\cdot)$ of $\zscr(\cdot)$. Therefore, by immediate application of
Lemma~\ref{lem:PoiLip}, we obtain that $\an(\cdot), a(\cdot)$ and $d(\cdot)$
are Lipschitz continuous. 

Finally, the Lipschitz continuity of $n(\cdot)$ and $q(\cdot)$ is established
if \FMSeq{enum:two} and \FMSeq{enum:five} hold --- this is because all of the
other components of $\zscr(\cdot)$ are Lipschitz continuous.  To this end,
recall that equations \eqref{eq:prim-f} and \eqref{eq:prim-q} are satisfied by
scaled system $\Zscr^{(r_k)}$ for all $t \in [0,T]$ by definition.  These
equation are preserved under the almost sure convergence $\Zscr^{(r_k)}(\cdot)
\to \zscr(\cdot)$.  Thus, \FMSeq{enum:two} and \FMSeq{enum:five} are
satisfied.

\medskip

\noindent{\em \FMSeq{enum:three}, \FMSeq{enum:four}, \FMSeq{enum:six}:}
These equations follow immediately by applying the later part of
Lemma~\ref{lem:PoiLip} for Poisson process (possibly time-varying)
$\An^{(r_k)}(\cdot), A^{(r_k)}(\cdot)$ and $D^{(r_k)}(\cdot)$.

\medskip

\noindent{\em \FMSeq{enum:seven}, \FMSeq{enum:eight}:}
Among remaining equations, first note that \FMSeq{enum:eight} follows because
$Z^{(r_k)}(\cdot), S^{(r_k)}(\cdot)$ and $\bar{X}^{(r_k)}(\cdot)$ are
non-decreasing and this property is preserved under the almost sure
convergence of $\Zscr^{(r_k)}(\cdot)$ to $\zscr(\cdot)$. A similar argument
establishes \FMSeq{enum:seven}.

\medskip

Now, consider a system that operates under the \MWUMa control policy. We wish
to establish that $\mu^{(\infty)}\big(\FMS^\alpha(T)\big) = 1$. This involves
further demonstrating that \FMSeq{enum:flow}--\FMSeq{enum:idling} are
satisfied.

\medskip

\noindent{\em \FMSeq{enum:flow}:} 
Consider any fixed flow type $f \in \Fscr$ and any regular point $t \in (0,T)$.
Now, if $n_f(t) = 0$, then it must be that $\dot{n}_f(t) = 0$. This is because
of the following argument, utilizing non-negativity of $n_f(\cdot)$: suppose
that either $\dot{n}_f(t) > 0$ or $\dot{n}_f(t) < 0$. Then, there exist times
$t^-<t$ or $t^+>t$ such that $n_f(t^{-}) < 0$ or $n_f(t^{+}) < 0$ --- this is
a contradiction. Given $\dot{n}_f(t) =
0$, by \FMSeq{enum:two} we have $\dot{\an}_f(t) = \dot{d}_f(t)$. By
\FMSeq{enum:three} and \FMSeq{enum:four}, it immediately follows that $x_f(t)
= \dot{\bar{x}}_f(t) = \nu_f/\mu_f$.

Now, suppose $n_f(t) > 0$. Recall that  $N_f^{(r_k)}(\cdot)$ converges
to $n_f(\cdot)$ as $k \to \infty$ under the Skorohod
metric. Therefore, it follows that there exists a $\delta > 0$ such that, for
$k$ sufficiently large, $N_f^{(r_k)}(s) \geq \delta$ for all $s \in [t-\delta,
t+\delta]$. We will consider $k$ to be sufficiently large for this to
hold. Since $t$ is a regular point, $\bar{x}_f(t)$ is differentiable. Consider
any $0<\varepsilon < \delta$.  Then, using the fact that $N_f^{(r_k)}(s) > 0$
for $s \in [t-\delta, t+\delta]$ and the radial invariance property of rate
allocation policy, we obtain
\begin{equation}\label{eq:fleq1}
  \begin{split}
    \bar{X}_f^{(r_k)}(t+\varepsilon) - \bar{X}_f^{(r_k)}(t) 
    & = \frac{1}{r_k} \int_{r_k t}^{r_k t+r_k\varepsilon} X_f(s)\, ds
    \\
    & = \frac{1}{r_k} \int_{r_k t}^{r_k t+r_k \varepsilon} 
    \left(
      \argmax_{x\in [0,C]}\ 
      \frac{x^{1-\alpha} \big(N_f^{(r_k)}\big)^\alpha}{1-\alpha} - 
      \big(Q_{\iota(f)}^{(r_k)}\big)^\alpha \, x\right)\, ds.
  \end{split}
\end{equation}
Define the function $\fR \colon (0,\infty)\times \Rp \to [0,C]$ by 
\[
\fR(n, q) \defeq \argmax_{x \in [0,C]}\ \frac{x^{1-\alpha} n^\alpha}{1-\alpha} - q^\alpha x.
\]
It can be easily checked that 
\[
\fR(n,q) = 
\begin{cases}
n/q & \text{if $n < Cq$,}
\\
C & \text{otherwise.}
\end{cases}
\]
Therefore, it follows
that $\fR$ is a continuous function. Further, $N_f^{(r_k)}(\cdot)$ and
$Q_{\iota(f)}^{(r_k)}(\cdot)$ are continuous as functions of time. Therefore,
treating $\fR\big(N_f^{(r_k)}(\cdot), Q_{\iota(f)}^{(r_k)}(\cdot)\big)$ as a
function of time, it is continuous and takes values in $[0,C]$. Over the
bounded interval $[t, t+\varepsilon]$, it must achieve a minimum and a
maximum, which we will denote by $\fR_{\min}(k, t, \varepsilon)$ and
$\fR_{\max}(k, t, \varepsilon)$, respectively. From \eqref{eq:fleq1}, it follows
that
\begin{equation}\label{eq:fleq2}
\fR_{\min}(k, t, \varepsilon)  \leq \frac{\bar{X}_f^{(r_k)}(t+\varepsilon) - \bar{X}_f^{(r_k)}(t)}{\varepsilon}  
\leq \fR_{\max}(k, t, \varepsilon). 
\end{equation}
Now, since $\big(\bar{X}^{(r_k)}(\cdot), N^{(r_k)}(\cdot), Q^{(r_k)}(\cdot))$
converges to $\big(\bar{x}(\cdot), n(\cdot), q(\cdot)\big)$ as $k \to \infty$,
it follows (due to the appropriate continuity of $\fR_{\min}$, and $\fR_{\max}$)
that
\begin{equation}\label{eq:fleq3}
\fR_{\min}(t, \varepsilon) 
\leq 
\frac{\bar{x}_f(t+\varepsilon) - \bar{x}_f(t)}{\varepsilon}  
\leq 
\fR_{\max}(t, \varepsilon).
\end{equation}
Here, $\fR_{\min}(t, \varepsilon)$ and $\fR_{\max}(t, \varepsilon)$ correspond to
the minima and maxima of $\fR\big(n(\cdot),q(\cdot)\big)$ over
$[t,t+\varepsilon]$.  Now, taking $\varepsilon \to 0$ in \eqref{eq:fleq3},
invoking the continuity of $\big(n(\cdot), q(\cdot)\big)$ and subsequently of
$\fR$, and recalling that $t$ is a regular point, we obtain
\[
x_f(t) = \dot{\bar{x}}_f(t)  =  \argmax_{x\in [0,C]}\ 
\frac{x^{1-\alpha} n_f^\alpha(t)}{1-\alpha} - 
q_{\iota(f)}^\alpha(t) \, x,
\]
when $n_f(t) > 0$.

\medskip

\noindent{\em \FMSeq{enum:packet}:}
Suppose $t\in (0,T)$ is a regular point and consider a schedule $\pi \in
\Sscr$. Assume there exists a schedule $\sigma\in\Sscr$ with
\begin{equation}\label{eq:fleq6}
\pi^\top (I - R) q^\alpha(t) 
< \sigma^\top (I - R) q^\alpha(t).
\end{equation}
We wish to establish that $\dot{s}_\pi(t) = 0$. Since $Q^{(r_k)}(\cdot)$
converges to $q(\cdot)$ as $k \to \infty$ the inequality \eqref{eq:fleq6} is
strict. It follows that there exists $\delta > 0$ such that for all $k$
sufficiently large and for all $s \in [t-\delta, t+\delta]$,
\[
\pi^\top (I - R) \big(Q^{(r_k)}(s)\big)^\alpha 
< 
\sigma^\top (I - R) \big(Q^{(r_k)}(s)\big)^\alpha.
\]
Then, for unscaled system, the weight of schedule $\pi$ is strictly less than
the weight of schedule $\sigma$ throughout the time-interval $[r_k (t-\delta),
r_k (t+\delta)]$.  Therefore, as per the \MWUMa scheduling policy, the
schedule $\pi$ is never chosen in this time period. That is, for the scaled
system, we have
\[
S^{(r_k)}_\pi(t-\delta) = S^{(r_k)}_\pi(t+\delta). 
\]
Then, as $k\to\infty$
\[
s_\pi(t+\delta) - s_\pi(t-\delta) = 0.
\]
Thus, $\dot{s}_\pi(t) = 0$.

\medskip

\noindent{\em \FMSeq{enum:nzero}:} 
Consider a regular point $t\in (0,T)$ with $n_f(t)=0$, for some
$f\in\Fscr$. By \FMSeq{enum:flow}, we have $x_f(t) = \rho_f$. Suppose that
$q_{\iota(f)}(t) > 0$.  For any $\varepsilon_1,\varepsilon_2>0$, there must exist $\delta>0$ so that, if $s
\in (t-\delta, t+\delta)$, 
\[
q_{\iota(f)}(s) \geq \varepsilon_1,
\quad
\text{and}
\quad
n_f(s) \leq \varepsilon_2.
\]
Therefore, if $k$ is sufficiently large, we must have 
\[
Q^{(r_k)}_{\iota(f)}(s) \geq \varepsilon_1/2,
\quad
\text{and}
\quad
N^{(r_k)}_f(s) \leq 2\varepsilon_2.
\]
Equivalently for the unscaled system,  
\[
Q_{\iota(f)}(r_k s) \geq r_k\varepsilon_1/2,
\quad
\text{and}
\quad
N_f(r_k s) \leq 2r_k\varepsilon_2.
\]
This that, for $k$ sufficiently  large and for all $s\in (t-\delta,t+\delta)$, the rate allocation in the unscaled system must satisfy
\[
X_f(r_k s) \leq \frac{4 \varepsilon_2}{\varepsilon_1}.
\]
This can be made smaller that $\rho_f/2$ by the appropriate choice of
$\varepsilon_2$, and this contradicts the fact that $x_f(t) =
\rho_f$. Therefore, it must be that $q_{\iota(f)}(t) = 0$.

\medskip

\noindent{\em \FMSeq{enum:idling}:} This follows in
a straightforward manner from the  invariant (for the unscaled system) that 
$Z(\tau) = \bzero$ for all $\tau\in\Zp$.

\end{document}